\documentclass[format=acmsmall, review=false]{acmart}
\pdfoutput=1
\usepackage{acm-ec-21}
\usepackage{booktabs} 
\usepackage[ruled]{algorithm2e} 

\SetAlFnt{\small}
\SetAlCapFnt{\small}
\SetAlCapNameFnt{\small}
\SetAlCapHSkip{0pt}
\IncMargin{-\parindent}

\usepackage{pgfplots}
\usepackage{tikz}
\pgfplotsset{compat=1.12}
\usetikzlibrary{calc}
\usetikzlibrary{patterns}

\setcitestyle{acmnumeric}

\usepackage{times}
\usepackage{soul}
\usepackage{url}
\usepackage{graphicx}
\usepackage{amsmath}
\usepackage{amsthm}
\usepackage{booktabs}
\usepackage{color}
\usepackage{mathrsfs}
\usepackage{verbatim}
\usepackage{amsfonts}
\usepackage{mathtools}
\mathtoolsset{showonlyrefs}
\usepackage{booktabs,microtype}
\usepackage{tikz}
\usetikzlibrary{arrows,shapes, calc, fit, positioning}
\usepackage{natbib}
\usepackage{graphicx}
\usepackage{xcolor}
\usepackage{comment}
\usepackage{enumitem}
\usepackage{tcolorbox}
\usepackage{xspace,bm}
\newtheorem{theorem}{Theorem}[section]
\newtheorem{definition}[theorem]{Definition}

\newtheorem{proposition}[theorem]{Proposition}
\newtheorem{corollary}[theorem]{Corollary}
\newtheorem{lemma}[theorem]{Lemma}

\newtheorem{example}[theorem]{Example}

\newcommand{\cLD}{\mathbf{cLD}} 
\newcommand{\sv}{\mathrm{sv}}

\DeclareMathOperator*{\argmin}{arg\,min}
\DeclareMathOperator{\supp}{supp}
\newcommand\restr[2]{{#1\mathord{|}_{#2}}}
\newcommand{\sfaithful}{subsidized-faithful\xspace} 
\newcommand{\smonotone}{subsidized-monotone\xspace} 
\newcounter{Bew1}
\newcounter{Bew2}

\begin{document}
\title{Fair and Truthful Mechanism with Limited Subsidy}
\author[H. Goko, A. Igarashi, Y. Kawase, K. Makino, H. Sumita, A. Tamura, Y. Yokoi, M. Yokoo]{Hiromichi Goko, Ayumi Igarashi, Yasushi Kawase, Kazuhisa Makino, Hanna Sumita, Akihisa Tamura, Yu Yokoi, Makoto Yokoo}

\begin{abstract}
The notion of \emph{envy-freeness} is a natural and intuitive fairness requirement in resource allocation. With indivisible goods, such fair allocations are unfortunately not guaranteed to exist. Classical works have avoided this issue by introducing an additional divisible resource, i.e., money, to subsidize envious agents. 
In this paper, we aim to design a truthful allocation mechanism of indivisible goods to achieve both fairness and efficiency criteria with a limited amount of subsidy. Following the work of Halpern and Shah,
our central question is as follows: to what extent do we need to rely on the power of money to accomplish these objectives? For general valuations, the impossibility theorem of combinatorial auction translates to our setting: even if an arbitrarily large amount of money is available for use, no mechanism can achieve truthfulness, envy-freeness, and utilitarian optimality simultaneously when agents have general monotone submodular valuations. By contrast, we show that, when agents have matroidal valuations, there is a truthful allocation mechanism that achieves envy-freeness and utilitarian optimality by subsidizing each agent with at most $1$, the maximum marginal contribution of each item for each agent. The design of the mechanism rests crucially on the underlying M-convexity of the Lorenz dominating allocations. 
\end{abstract}

\begin{titlepage}

\maketitle

\end{titlepage}

\section{Introduction}
Consider a group of employees with preferences over their shifts; some may prefer to work in the morning, whereas others may prefer to work in the afternoon. All employees are willing to work, but they may differ in the extent to which they like each time slot. How can shifts be scheduled such that the resulting allocation is fair among employees? This question falls under the realm of the fair division problem, whereby indivisible resources are distributed among heterogeneous participants. 

The notion of fairness that has been extensively studied in the literature is {\em envy-freeness} \cite{Foley}. It requires that no agent wants to swap their bundle with that of another agent. When the resource to be allocated is divisible, the classical result ensures the existence of an envy-free allocation \cite{Varian}; when the resource is indivisible, envy-freeness is not a reasonable goal. A relevant example is the case of one item and two agents: no matter how we allocate the single item, the agent who gets nothing envies the other. Hence, the only “fair” solution is to give nothing to both agents.

One way to circumvent this issue is monetary compensation. Money is a powerful tool used to incentivize people and to accomplish fairness desirably. In the preceding example, the employer may attempt to balance inequality, e.g., by compensating employees who are assigned to the night shifts. Another example is the governmental body that subsidizes health workers in rural and remote areas. 

In mechanism design with money, envy-freeness can indeed be achieved
by the well-known Vickrey–Clarke–Groves (VCG) auction mechanism
\cite{clarke,groves:econometrica:1973,vickrey:1961} in cases where
when each agent's valuation is superadditive \cite{papai:scw:2003}; in principles, this mechanism is guaranteed to be
envy-free, truthful, and utilitarian optimal if one allocates enough amount of money to participants\footnote{We will formalize this argument in Section \ref{sec:superadditive}.}
assuming each agent's valuation is superadditive. 
In several applications,
however, the resulting outcome of VCG may be unsatisfactory in the
following three respects. 
First, the social planner may have a limited amount of money that can be used to subsidize the participants; for example, employees are usually paid additional compensation up to some limit.
Second, the allocation itself only aims to
maximize the utilitarian social welfare, ignoring the requirement of fairness; in an extreme case, one agent who happens to have higher utilities may get all the resources, whereas others get nothing. 
Third, when some agent has a non-superadditive valuation, VCG fails to satisfy envy-freeness. 

In this paper, we study the allocation mechanisms of indivisible items with limited subsidies. 
Formally, we work in the setting of Halpern and Shah \cite{HalpernShah2019}. There, a set of indivisible items together with subsidies are to be distributed among agents who have quasi-linear preferences. The objective is to bound the amount of subsidies necessary to accomplish envy-freeness, assuming that the maximum value of whole items is at most the number $m$ of items for each agent.\footnote{Halpern and Shah \cite{HalpernShah2019} dealt with additive valuations and assumed that the maximum value of each single item is at most $1$; our works deals with valuations that are not necessarily additive, and assumes a more general condition--, i.e., that $v_i(M) \leq m$ for each agent $i$.} Although Halpern and Shah \cite{HalpernShah2019} and subsequent works \cite{Brustle2020,caragiannis2020computing,Aziz2020} are mostly concerned with fairness criteria, we take the mechanism-design perspective: in practice, agents may behave strategically rather than truthfully when reporting their preferences. The goal of this paper is to analyze the amount of subsidies required to accomplish the three basic desiderata of a mechanism: truthfulness, envy-freeness, and utilitarian optimality. 

\subsection{Our contributions}
We will first focus upon a class of valuations that exhibit \emph{substitutability}, which is an essential characteristic that is common in many practical situations. 
For example, employees would like to work in some time slots, but working all day long is not preferable because of overwork. To capture such a phenomenon, it is natural to consider the class of submodular valuations. 
Although the impossibility result of combinatorial auction immediately applies to these valuations, our question is whether there is any {well-structured} subclass of submodular valuations that guarantees a desired mechanism. 

A subclass of monotone submodular valuations that arises in several applications is that of matroidal valuations, i.e., submodular functions with dichotomous marginals. For example, suppose that employees are allocated to tasks of various types; they may either approve or disapprove of each task depending on their abilities and can perform certain combinations of tasks under hierarchical constraints (e.g., the number of tasks that can be assigned to each employee is determined per morning/afternoon or day/week/month). One can naturally express such situations by the matroidal valuations associated with laminar matroids (see, e.g., \cite{fife2017laminar}). 
Another application is when a social planner desires to allocate public housing to people in a way that is fair across different social/ethnic groups \cite{benabbou}; this situation can be modeled by binary assignment valuations, which belong to a class of matroidal valuations: each group's satisfaction is set to be the optimal value of assignments of items to group members. 

The class of matroidal valuations turns out to be fruitful in the standard setting of fair division without subsidy \cite{babaioff2020fair,benabbou2020,halpern2020fair}; particularly, they do admit an allocation rule that is truthful, approximately fair, and efficient. Babaioff et al.~\cite{babaioff2020fair} very recently designed such a mechanism, called the prioritized egalitarian (PE) mechanism. With ties broken according to a prefixed ordering over the agents, the mechanism returns a \emph{clean Lorenz dominating allocation}, i.e., an allocation whose valuation vector (weakly) Lorenz dominates those under the other allocations and whose bundles include no redundant items\footnote{Items that can be removed without decreasing the agents' valuations.}. 

Now, returning to our setting, is it possible to design a desired mechanism with a limited amount of subsidy? In Section \ref{sec:matroidal}, we show that, when agents have matroidal valuations, there is a polynomial-time implementable truthful mechanism that achieves both envy-freeness and utilitarian optimality, with a subsidy of at most $1$ for each agent. Note that a mere extension of the previously known mechanism \cite{halpern2020fair,babaioff2020fair} does not achieve these properties; informally, by distributing the commonly desirable good, some agents may be incentivized to pretend that they envy others although they do not.\footnote{We will formally discuss details in Example \ref{ex:lying}.} 

To prevent agents from misreporting, we design a polynomial-time implementable mechanism, the so-called \emph{subsidized egalitarian} (SE) mechanism, which resembles the classic VCG in the sense that it \emph{punishes} agents who may potentially decrease others' valuations. At a high-level, the mechanism hypothetically distributes $1$ dollar to each agent upfront and implements the auction over the set of {clean Lorenz dominating allocations} ($\cLD$). By contrast with the PE mechanism, the actual allocation can be taken \emph{arbitrarily} from these allocations; then, each agent who benefits from the allocation pays $1$ dollar back to the mechanism designer. 
Note that the total amount of subsidies required by the SE mechanism is at most $n-1$, which cannot be improved as the worst-case guarantee.\footnote{Consider one item and $n$ agents. If all agents desire the single item, subsidy $1$ must be given to every agent but the one who gets the item to achieve envy-freeness.} A further, perhaps surprising, remark is that, in the output of the SE mechanism, the utility of each agent \emph{does not} change depending on the choice of an allocation. 

Technically, we build an essential connection between discrete convex analysis and fair division. Indeed, our result crucially rests on the underlying structural property of clean Lorenz dominating allocations; when agents have matroidal valuations, the valuation vectors associated with the $\cLD$ enjoy matroidal M-convexity, allowing one to obtain a new clean Lorenz dominating allocation from another via natural exchange operations. 

In Section~\ref{sec:withoutfreedisposal}, we further discuss the setting without the free-disposal assumption. It is often assumed that the mechanism can throw away any part of the resource and may thus leave some items unallocated; however, there are several practical scenarios wherein this is not ideal. Examples include allocating shifts to medical workers and assigning tasks to employees. Unfortunately, we observe that, even when agents have binary additive valuations, no truthful and envy-free mechanism allocates all items and returns a Lorenz dominating allocation with each agent being subsidized by at most $1$. 
However, dropping the truthfulness requirement, we show that there is a polynomial time algorithm that accomplishes envy-freeness and utilitarian optimality while each agent is subsidized at most $1$ and all items are allocated to some agent for matroidal valuations. Of independent interest, we also prove in the Appendix that the resulting allocation of the algorithm satisfies an approximate fairness notion, called envy-freeness up to any good (EFX). 

If agents are broadly expressive(i.e., the family of agents' valuations satisfies the so-called \emph{convexity} condition\footnote{See Definition $1$ of \cite{holmstrom:econometrica:1979}.}), Groves mechanisms are known to be the unique family of mechanisms that satisfy truthfulness and utilitarian optimality \cite{holmstrom:econometrica:1979}. Hence, our hopes are centered on such mechanisms for a rich class of valuations. 
Although VCG fails to satisfy envy-freeness for monotone submodular valuations \cite{Feldman2012}, P\'{a}pai \cite{papai:scw:2003} showed that it is envy-free when agents have superadditive valuations--i.e., when agents' preferences do not exhibit substitutability. 
These results have immediate implications for our setting. In Section \ref{sec:superadditive}, we show that, for superadditive valuations, there is a truthful mechanism that achieves envy-freeness and utilitarian optimality, with each agent receiving a subsidy of at most $m$; furthermore, we show that the amount $m$ is necessary even when agents have additive valuations. In Section \ref{sec:general}, we observe that, even if an arbitrarily large amount of money is available for use, no mechanism can achieve truthfulness, envy-freeness, and utilitarian optimality simultaneously. 

Figure~\ref{fig:valuation} illustrates the relationship among classes of valuation functions.

\vspace{2em}
\begin{figure}[ht]
		\centering
		\begin{tikzpicture}[xscale=0.8, yscale=0.8, transform shape,every node/.style={minimum size=6mm, inner sep=1pt}]
				\draw[rotate=-5,fill=gray!10] (-2,0) ellipse (3.8cm and 2.5cm);
				\draw[gray,ultra thick,fill=gray, fill opacity=0.6] (1,-0.2) ellipse (2.6cm and 1.2cm);
				\draw[rotate=5] (1,0) ellipse (3.5cm and 2.2cm);

				\node[opacity=0.75] at (0,-0.5){\bf Binary Additive};
				\node at (0,1.5){\bf Additive};
				\node at (1.2,0.5){\Large \bf Matroidal};
				\node at (2.6,1.6){\bf Submodular};
				\node at (-3.8,1.6){\bf Superadditive};
				\end{tikzpicture}

\caption{Classes of valuation functions with which we deal in this paper. The SE mechanism applies to the class of matroidal valuations (dark gray area), whereas VCG applies to the class of superadditive valuations (light gray area). }
\label{fig:valuation}
\end{figure}
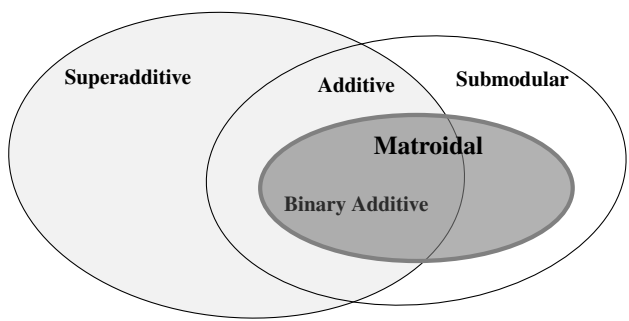

\subsection{Related work}
The idea of compensating an indivisible resource allocation with money has been prevalent in classical economics literature \cite{Alkan1991,Maskin1987,Klijn,moulin2004fair,SunYang2003,Svensson1983,Tadenuma1993}; much of the classical work has focused on the unit-demand case in which each agent is allocated at most one good. Examples include the famous rent-division problem of assigning rooms to several housemates and dividing the rent among them \cite{SuRentalHarmony}. It is known that, for sufficient amount of subsidies, an envy-free allocation exists \cite{Maskin1987} and can be computed in polynomial time \cite{Aragones,Klijn}. 

Most classical literature, however, has not considered a situation in which the number of items to be allocated exceeds the number of agents, in contrast to the rich body of recent literature on the multi-demand fair division problem. 
Halpern and Shah \cite{HalpernShah2019} recently extended the model to the multi-demand setting wherein each agent can be assigned to multiple items. 
Despite the existence of numerous related papers, \cite{HalpernShah2019} is the first work to study the asymptotic bounds on the amount of subsidy required to achieve envy-freeness. 
Halpern and Shah \cite{HalpernShah2019} showed that, for binary additive valuations, an allocation that maximizes the Nash welfare (MNW) can be made envy-free with a subsidy of at most $1$ for each agent. They further conjectured that, for general additive valuations in which the value of each item is at most $1$, giving at most $1$ to each agent is sufficient to eliminate envies. Brustle et al.~\cite{Brustle2020} affirmatively settled this conjecture by designing an algorithm that iteratively solves a maximum-matching instance; note that our work is the first to show that for valuations that are not necessarily additive, envy-freeness and completeness can be accomplished by giving each agent at most $1$ subsidy. 

Caragiannis and Ioannidis \cite{caragiannis2020computing} studied the computational complexity of approximating the minimum amount of subsidies required to achieve envy-freeness. 
Aziz \cite{Aziz2020} considered another fairness requirement, the so-called \emph{equitability}, in conjunction with envy-freeness and characterized an allocation that can be made both equitable and envy-free with money.

Closely related to the present study are the works of \cite{babaioff2020fair,benabbou2020}, who study the fair allocation of indivisible items with matroidal valuations. The PE mechanism of \cite{babaioff2020fair} returns an allocation that maximizes the Nash welfare and achieves envy-freeness up to any good (EFX) and utilitarian optimality. Benabbou et al.~\cite{benabbou2020} focused more upon the balance between efficiency and fairness. They showed that when agents have matroidal valuations, leximin allocations are equivalent to MNW allocations, which-- together with the result of Babaioff et al.~\cite{babaioff2020fair}--implies that MNW allocations are Lorenz dominating. 


The case where (possibly) multiple items can be allocated to each agent while the agents
pay some amount of money to the mechanism designer (or a special agent
called seller), is extensively studied in combinatorial auctions
\cite{cramton:2005}. A representative mechanism is the well-known VCG
mechanism, which is truthful and utilitarian optimal.
Envy-freeness is not a central issue in combinatorial auctions,
with a notable exception presented by P\'{a}pai \cite{papai:scw:2003}.
The results obtained by P\'{a}pai \cite{papai:scw:2003} can be applied
to cases wherein each agent receives a non-negative amount of money
(subsidy). 

Our paper is at the intersection of discrete convex analysis and economics.
Recent advances in discrete convex analysis have found a variety of applications
in economics, including exchange economies with indivisible goods
\citep{Murota:SIAM:2003,MT:dca:2003,SY:dca:2013},
inventory management
\citep{Zipkin:Inventory:2008,Huh:Inventory:2010}, auctions
\citep{MSY:auction:2013},
and two-sided matching \citep{murota:metr:2013,kty:2018}.\footnote{%
See the survey paper by Murota~\cite{murota:dca:2016}.}
As this long, but incomplete list would suggest,
techniques from this literature are useful for a wide variety
of problems. We add fair division problems with a limited subsidy to this list.

\section{Model}\label{sec:model}
We model fair division with a subsidy as follows.
For $k \in \mathbb{N}$, we denote $[k]=\{1,\ldots,k\}$.
Let $N=[n]$ be the set of given $n$ agents and let $M=[m]$ be the set of given $m$ indivisible goods.
Each agent $i$ has a \emph{valuation function} $v_i : 2^M \rightarrow \mathbb{R}_+$ with $v_i(\emptyset)=0$, where $\mathbb{R}_+$ is the set of non-negative reals.
For notational simplicity, we write $v_i(e)$ instead of $v_i(\{e\})$ for all $e \in M$.
In this paper, we assume that valuation functions are \emph{monotone}:
$v_i(X) \leq v_i(Y)$ for any $X \subseteq Y \subseteq M$.

\subsubsection*{Valuation functions}
We focus upon the following classes of valuation functions:
\begin{description}
\item[General:] We assume that the maximum valuation is bounded, i.e.,
  $v_i(M) \leq m$ holds for all $i \in N$;
    
\item[Superadditive:] A special case of the general valuations, 
  where for any $i \in N$, $X, Y \subset M$, s.t.
  $X \cap Y = \emptyset$, $v_i(X) + v_i(Y) \leq v_i(X\cup Y)$ holds; 
  
\item[Additive:] A special case of the superadditive valuations, where 
  $v_i(X)=\sum_{e\in X} v_i(e)$ holds for any $X \subseteq M$, $i\in N$;
  
\item[Binary additive:] A special case of the additive valuations, where 
  $v_i(e)\in \{0,1\}$ for any $e \in M$ and $i \in N$;

\item[Matroidal:] A super class of the binary additive valuations, where 
  (i) the marginal contribution $v_i(X\cup \{e\})-v_i(X)$ is either 0 or 1 for all $X \subsetneq M$ and $e\in M\setminus X$, and 
  (ii) $v_i$ is submodular, i.e., $v_i(X)+v_i(Y) \geq v_i(X\cup Y)+v_i(X\cap Y)$ holds for all $X, Y \subseteq M$.
\end{description}

We remark that a matroidal valuation function is a rank function of a matroid\footnotemark. 
For a matroidal valuation $v_i$, each set $X \subseteq M$ such that $v_i(X)=|X|$ is called an \emph{independent set}.

\footnotetext{Let $E$ be a finite set and let $r: 2^E \rightarrow \mathbb{Z}_+$. A set system $(E, r)$ is a \emph{matroid} if for all $X, Y \subseteq E$, 
(i) $X \subseteq Y \Rightarrow$ $r(X) \leq r(Y) \leq |Y|$, and 
(ii) $r(X)+r(Y) \geq r(X\cup Y)+r(X\cap Y)$. 
If $(E, r)$ is a matroid, then $r$ is called a \emph{rank} function.}

\subsubsection*{Allocations}
An allocation of goods is an ordered subpartition of $M$ into $n$ bundles. 
We denote an allocation by $A=(A_1, \ldots, A_n)$ such that $A_i \subseteq M$ for all $i \in N$ and $A_i \cap A_j = \emptyset$ for any $i \neq j$. 
In allocation $A$, agent $i$ receives a bundle $A_i$ of goods.
We will deal with two types of allocation: (1) a \emph{complete} allocation (that is, every good must be allocated to some agent), and (2) an incomplete allocation (that is, we can leave some goods unallocated).


We introduce notions of efficiency that we use in this paper.
We say that $A$ is \emph{Pareto optimal} if there exists no allocation $A'$ such that $v_i(A_i)\leq v_i(A'_i)$ for any $i\in N$ and $v_i(A_i) < v_i(A'_i)$ for some $i\in N$.
The {\em utilitarian social welfare} of an allocation $A$ is $\sum_{i\in N} v_i(A_i)$, and $A$ is a \emph{utilitarian optimal} allocation if it maximizes the utilitarian social welfare among all allocations. 
A refinement of utilitarian optimality is Lorenz dominance: given allocations $A$ and $B$, we say that $A$ \emph{Lorenz dominates} $B$ if, for every $k\in[n]$, the sum of the smallest $k$ values in $(v_1(A_1),\ldots, v_n(A_n))$ is at least as large as that of $(v_1(B_1),\ldots, v_n(B_n))$, i.e., if $v_{i_1}(A_{i_1})\leq \dots \leq v_{i_n}(A_{i_n})$ and $v_{j_1}(B_{j_1})\leq \dots \leq v_{j_n}(B_{j_n})$ (where $\{i_1,\dots,i_n\}=\{j_1,\dots,j_n\}=[n]$), then $\sum_{\ell=1}^k v_{i_\ell}(A_{i_\ell}) \geq \sum_{\ell=1}^k v_{j_\ell}(B_{j_\ell})$ holds for each $k$.
A \emph{Lorenz dominating allocation} is an allocation that Lorenz dominates every other allocation.
The following proposition holds from the definition of Lorenz dominance with $k=n$.
\begin{proposition}
Every Lorenz dominating allocation is utilitarian optimal. 
\end{proposition}

Lorenz dominance is also an egalitarian fairness notion in the sense that the least happy agent becomes happier to the greatest extent possible.
Another allocation that often achieves the sweet spot of efficiency and fairness is a {maximum Nash welfare (MNW)} allocation \cite{CKM+16a,babaioff2020fair,benabbou2020}. 
We say that $A$ is a \emph{maximum Nash welfare (MNW)} allocation if it maximizes the number of agents receiving positive utility and, subject to that, maximizes the product of the positive utilities, i.e.,  $\prod_{i\in N:\, v_i(A_i)>0} v_i(A_i)$.

It is known that for matroidal valuations, MNW allocations coincide with Lorenz dominating allocations\footnote{More generally, the set of Lorenz dominating allocations is equivalent to that minimizing a symmetric strictly convex function~\cite{FM19,benabbou2020}.} \cite{benabbou2020,babaioff2020fair}.
%
%
Note that a Lorenz dominating allocation always exists for matroidal valuation functions, whereas it may not exist in general.

To find efficient allocations, it is often necessary to avoid redundancy in allocations.
An allocation $A$ is called {\em clean} if $v_i(A_i\setminus\{e\})<v_i(A_i)$ for any $i\in N$ and $e\in A_i$. 
Note that any allocation can be transformed into a clean one without changing valuations by removing items of zero marginal gain from respective agents.
For matroidal valuations, $A$ is clean if and only if $v_i(A_i)=|A_i|$ for each $i \in [n]$ (see also~\cite{benabbou2020}); thus, we see that for matroidal valuations, an allocation $A$ is clean Lorenz dominating if and only if for every clean allocation $B$, the total size of the smallest $k$ bundles in $A$ is at least as large as that of $B$ for each $k\in[n]$.
This fact will be used in Section~\ref{sec:matroidal}.

\subsubsection*{Fairness with a subsidy}
Our goal is to achieve an envy-free allocation of indivisible goods using a limited amount of \emph{subsidy}, which is an additional divisible good. 
We denote by $p=(p_1, \ldots, p_n)\in \mathbb{R}_+^n$ a subsidy vector, whose $i$th entry $p_i$ is the amount of subsidy received by agent $i$. 
For allocation $A$ and a subsidy vector $p$, we call $(A,p)$ an allocation with a subsidy; we assume that each agent has a standard quasi-linear utility, i.e., the utility of agent $i$, who obtains a bundle $X$ and subsidy $p_i$, is equal to: $v_i(X) + p_i$. 
The envy-freeness for an allocation with a subsidy is defined as follows:
\begin{definition}
An allocation with a subsidy $(A, p)$ is \emph{envy-free} if $v_i(A_i)+p_i \geq v_i(A_j)+p_j$ for all agents $i,j\in N$.
\end{definition}

An allocation $A$ is called {\em envy-freeable} if there exists a subsidy vector $p$ such that $(A, p)$ is envy-free.
Halpern and Shah \cite{halpern2020fair} characterized envy-freeable allocations using envy graphs defined as follows.
For an allocation $A$, its envy graph $G_A$ is the complete weighted directed graph whose node set is the agent set $N$; for each $i,j\in N$, the arc $(i,j)$ has weight $w(i,j)=v_i(A_j)-v_i(A_i)$,
which represents the amount of envy of $i$ towards $j$. 
This value can be negative if $i$ prefers their bundle to $j$'s bundle. 
A {\em walk} $Q$ in $G_A$ is a sequence of nodes $(i_1,i_2,\dots,i_k)$, and its weight is defined as
$w(Q)=\sum_{t=1}^{k-1}w(i_t,i_{t+1})$. 
A walk is a {\em path} if all nodes are distinct, and a {\em cycle}
if $i_1,i_2,\dots,i_{k-1}$ are all distinct and $i_1=i_k$.
The following theorem combines Theorems~1 and 2 of \cite{halpern2020fair}.
\begin{theorem}[Halpern and Shah \cite{halpern2020fair}]\label{thm:envy-graph}
For any allocation $A=(A_1,\dots,A_n)$ and any $q\in \mathbb{R}_+$, the following two are equivalent:
\begin{itemize}
    \item[(a)] $A$ is envy-freeable with a subsidy of at most $q$ for each agent.
    \item[(b)] $G_A$ has neither a positive-weight cycle nor a path with a weight larger than $q$.
\end{itemize}
When (b) holds, if we set $p_i$ as the maximum weight of any path starting at $i$ in $G_A$ for each $i\in N$, then $(A,p)$ is envy-free. 
\end{theorem}
The equivalence of envy-freeability and the nonexistence of positive-weight cycles in $G_A$
is shown in Theorem~1 of Halpern and Shah~\cite{halpern2020fair}.
The relationship between the bound of subsidy and the maximum weights of paths in $G_A$ follows from their Theorem~2.

\subsubsection*{Mechanisms}
In each subsequent section, we assume that a valuation function of each agent is taken from some specified function class $V$. 
For example, in Section~\ref{sec:matroidal}, we let $V$ be the set of all matroidal functions on $M$.
A {\em valuation profile}, or just a {\em profile}, is a tuple $(v_1,\dots,v_n)\in V^N$ of the valuation functions of the all agents in $N$.
For resource allocation with a subsidy, a {\em mechanism} is a mapping of a valuation profile to an allocation with a subsidy.
A mechanism first asks each agent to report a valuation function and then outputs an allocation with subsidy on the basis of the reported valuations. We notice that the reported valuations may be different from the true ones. 

Some agents may have incentives to report a false valuation function to obtain a larger utility.
To prevent such manipulation, truthfulness is a standard requirement for mechanisms. 
A mechanism is \emph{truthful} if reporting the true valuation function maximizes the agent's utility, given the fixed reports of the other agents. A more precise definition is as follows:
for every agent $i$, every profile $(v_1,\dots,v_n)\in V^N$, and every $v'_i\in V$, 
if we denote by $(A, p)$ and $(A',p')$ the outputs of the mechanism for the profiles $(v_1,\dots,v_i,\dots,v_n)$ and $(v_1,\dots v'_i,\dots,v_n)$, respectively, then $v_i(A_i)+p_i\geq v_i(A'_i)+p'_i$.

We say that a mechanism satisfies property P if it outputs an allocation satisfying P.
For example, a mechanism satisfies envy-freeness if it outputs an envy-free outcome, and similarly for other properties such as MNW, completeness, and utilitarian optimality.

\section{Matroidal valuations}\label{sec:matroidal}
In previously explained applications such as shift scheduling, goods usually have substitute properties; therefore, we are interested in the setting with submodular valuation functions. For such a setting, can we design a mechanism that simultaneously achieves truthfulness, efficiency, and fairness with small amount of subsidies?

Generally, the impossibility result of the combinatorial auction applies to monotone submodular valuations \cite{Feldman2012}; we are, however, able to answer this question affirmatively for the class of matroid rank valuations, i.e., submodular functions with dichotomous marginals. By setting the domain $V$ of the valuation functions as matroidal functions, we can show that giving at most $1$ subsidy to each agent suffices to accomplish these goals. Note that such a mechanism has not been shown to exist even for binary additive valuations. Our main theorem in this section is stated as follows: 


%

\begin{theorem}\label{thm:matroidrank:truthful}
For matroidal valuations, there is a polynomial-time implementable mechanism that is truthful, utilitarian optimal, and envy-free with each agent receiving subsidy $0$ or $1$, and the the total subsidy being at most $n-1$.
\end{theorem}

%
Before presenting a mechanism, let us illustrate how the problem can become tricky even when agents have binary additive valuations: The following simple example shows that we have to give some subsidy to agents who want nothing.  

\begin{example}\label{ex:lying}
Consider two agents $N=\{1,2\}$ and one item $M=\{e\}$ with each agent either wanting the item or not (i.e., valuation for the item is either $0$ or $1$). Suppose that there is a mechanism that is truthful, envy-free, and utilitarian optimal. Consider two profiles $P_1$ and $P_2$. 
In $P_1$, both agents desire the single item. In this case, the outcome must be such that one agent receives the item and the other receives nothing. Without loss of generality, we assume that agent $1$ receives the item. By envy-freeness, agent $2$ must obtain at least $1$ subsidy. 
In $P_2$, agent $1$ reports that she wants the item but agent $2$ does not; then, the item must be allocated to agent $1$ who desires the item by utilitarian optimality. 
Now, it appears that no subsidy is needed in $P_2$ because agents do not envy each other. However, it turns out that we \emph{do} have to subsidize the agent who wants nothing; otherwise, agent $2$ would benefit by misreporting that she wants the item. 
\end{example}

%

Note that, for a Lorenz dominating allocation, we can easily compute the amount of subsidy required to make it envy-free by Theorem~\ref{thm:envy-graph}; however, as we observed in Example~\ref{ex:lying}, 
the mechanism should account for an exponential number of profiles if it aims to compute the minimum amount of additional subsidies to achieve truthfulness.
Rather, the mechanism “generously” distributes subsidies.

Our mechanism, which we refer to as \emph{subsidized egalitarian} (SE), proceeds as follows. First, it arbitrarily chooses a clean Lorenz dominating allocation that coincides with a clean MNW and is thus guaranteed to exist under matroidal valuations~\cite{babaioff2020fair,benabbou2020}; then, it subsidizes agents with the following condition: the valuation of allocated bundle is (i) the same as the worst (clean) Lorenz dominating allocation and (ii) not the largest among the agents. The mechanism thus ensures that the utility of agent $i$ is equal to the valuation of the worst clean Lorenz dominating allocation plus $1$ if she is not the one who receives the largest bundle.

Recall that, for matroidal valuations, allocation $A$ is clean if and only if $v_i(A_i)=|A_i|$ for any $i\in N$. 
For a profile $P=(v_1,\dots,v_n)$, let $\cLD[P]$ be the set of clean Lorenz dominating allocations.
To ease notation, we often omit the argument $P$ if no confusion will arise. 
Formally, our mechanism is summarized as follows.



\begin{tcolorbox}[title=Subsidized Egalitarian, left=0mm]
\begin{enumerate}[label=\textbf{Step \arabic*.},leftmargin=*]
    \item Allocate items according to an arbitrarily chosen $A \in \cLD$.
    \item Give $1$ subsidy to each $i\in N$ if (i)~$|A_i|=\min_{B\in\cLD}|B_i|$ and (ii)~$|A_i|<\max_{j\in N}|A_j|$.
\end{enumerate}
\end{tcolorbox}
The mechanism returns a utilitarian optimal allocation according to the property of Lorenz dominating allocations. Clearly, the subsidy for each agent is $0$ or $1$. The total subsidy is at most $n-1$ since at least one agent (who receives $\max_{j\in N}|A_j|$ items) gets no subsidy.
Remarkably, we observe that the difference between the valuations of the best and the worst Lorenz dominating allocations is at most one for every agent (Proposition~\ref{prop:01}) and that 
the utility of each agent does not change with the choice of an allocation in Step~1 (Proposition~\ref{prop:SE_uu}).
Hence, the utility of each agent is at least the valuation of the best clean Lorenz dominating allocation.

Here, we note that the SE mechanism imposes the condition (ii)~$|A_i|<\max_{j\in N}|A_j|$ in Step 2 to avoid giving all agents subsidy $1$. In fact, a variant of the SE mechanism in which condition (ii)~is removed 
fulfills all properties required by Theorem~\ref{thm:matroidrank:truthful} except that the total subsidy is at most $n$, instead of $n-1$.
It is also worth noting that, without subsidy, simply picking an arbitrary allocation in $\cLD$ does not guarantee truthfulness \cite[Example 4]{babaioff2020fair}.

We will prove that the SE mechanism satisfies the desired properties in Theorem~\ref{thm:matroidrank:truthful} through the following steps. First, we will provide the structural properties of $\cLD$ in Section~\ref{subsec:structure} and prove that the SE mechanism is polynomial time implementable in Lemma~\ref{lem:SE_poly}.
We will further show that the mechanism is envy-free and truthful in Lemmas~\ref{lem:SE_EF} and~\ref{lem:SE_truthful}, respectively. Throughout, we assume that all agents have matroidal valuations. 

\subsection{Structure of clean Lorenz dominating allocations}\label{subsec:structure}
As a preparation for the proof of Theorem~\ref{thm:matroidrank:truthful}, we introduce some notations. For an allocation $A$, let $\sv(A)$ be a size vector $(|A_1|,\dots,|A_n|)$ and let $\sv^\uparrow(A)$ be a vector obtained from $\sv(A)$ by rearranging its components in increasing order.
Recall that a clean allocation $A$ is Lorenz dominating if and only if for each clean allocation $B$, it holds that $\sum_{i=1}^k\sv^\uparrow(A)_i\ge \sum_{i=1}^k\sv^\uparrow(B)_i$ for each $k\in [n]$. Note that $\sv^\uparrow(A)$ is unique across all $A\in\cLD[P]$ according to the definition of $\cLD[P]$.

For any finite set $E$ and any $i\in E$, a characteristic vector $\chi_{i}$ is an $E$-dimensional vector whose $i$-th entry is $1$ and whose other entries are all $0$.
For two vectors $x, y\in \mathbb{Z}^E$, we define $\supp^+(x-y)\coloneqq \{i\in E\mid x(i)>y(i)\}$ and $\supp^-(x-y)\coloneqq\{i\in E\mid x(i)<y(i)\}$.
For a valuation function $v_i$ and $X \subseteq M$, a set function $\restr{v_i}{X}$ defined as $\restr{v_i}{X}(Y)=v_i(X\cap Y)$ for all $Y\subseteq M$ is called a restriction of $v_i$ to $X$.

Recall that, for a matroidal valuation function $v_i$, a subset $X\subseteq M$ is called {\em independent} if $v_i(X)=|X|$. The family of independent sets of any matroidal function is known to satisfy the following {\em augmentation property}: if both $X$ and $Y$ are independent and $|X|<|Y|$, then there exists an item $e\in Y\setminus X$ such that $X\cup\{e\}$ is also independent. A maximal independent set is called a {\em base}; by the augmentation property, all bases have the same cardinality.
\smallskip

We first present the following lemma, shown in the proof of \cite[Lemma 17]{babaioff2020fair}, concerning an operation that moves an allocation closer to another allocation in the terms of size vectors.
Note that this operation can be interpreted as an augmenting path in the exchange graph of a matroid intersection (see, e.g., \cite{Schrijver2003} for details).
For an allocation $A$, we use $A_0$ to denote the set of unallocated items $M\setminus\bigcup_{i\in N}A_i$.
\begin{lemma}\label{lem:transfer}
Let $A$ and $B$ be two clean allocations, and let $i$ be an agent.
If $|A_i|>|B_i|$, there exists a sequence of clean allocations $C^0,C^1,\dots,C^r$ with the following properties:
\begin{enumerate}[label={\rm(\roman*)}]
\item $C^0=B$, $k^0=i$,
\item $e^{t}$ is an item such that $e^{t}\in A_{k^{t-1}}\setminus C^{t-1}_{k^{t-1}}$ and $C^{t-1}_{k^{t-1}}\cup\{e^t\}$ is independent for $v_{k^{t-1}}$ ($t=1,\dots,r$),\label{condition:item}
\item $k^t\in N\cup\{0\}$ is the index such that $e^t\in C^{t-1}_{k^t}$ ($t=1,\dots,r$),
\item $C^{t}$ is the allocation that is obtained from $C^{t-1}$ by transferring $e^{t}$ from $k^{t}$ to $k^{t-1}$ ($t=1,\dots,r$),
\item $|A_{k^r}|<|B_{k^r}|$.
\end{enumerate}
\end{lemma}
\begin{proof}
We can find the sequence $C^0,C^1,\dots,C^r$ by arbitrarily selecting the items $e^1,e^2,\dots,e^r$ satisfying condition \ref{condition:item} of Lemma~\ref{lem:transfer} iteratively.
In fact, if $|A_{k^{t-1}}|\ge |B_{k^{t-1}}|$, we have $|A_{k^{t-1}}|>|C^t_{k^{t-1}}|$ because $|B^t_{k^{t-1}}| > |C^t_{k^{t-1}}|$; hence, there exists $e^{t}$ such that
$e^{t}\in A_{k^{t-1}}\setminus C^{t-1}_{k^{t-1}}$ and $C^{t-1}_{k^{t-1}}\cup\{e^{t}\}$ is independent for $v_{k^{t-1}}$ by the matroid augmentation property.
This procedure terminates in a finite number of steps because $\sum_{i\in N}|A_i\bigtriangleup C_i^t|$ is strictly monotone decreasing with respect to $t$.
\end{proof}
Note that $\sv(C^t)=\sv(C^0)+\chi_{k^0}-\chi_{k^t}$ if $k^t\in N$; additionally, if allocation $B$ is utilitarian optimal, then $k^r$ must be in $N$.

A key structure of $\cLD$ is the \emph{M-convex} structure of size vectors. A non-empty set $S\subseteq\mathbb{Z}^E$ is said to be \emph{M-convex} if it satisfies the following (simultaneous) \emph{exchange property}:
\begin{description}
\item[(B-EXC)]
For any $x,y\in S$ and for any $i\in\supp^+(x-y)$, there exists some $j\in\supp^-(x-y)$ such that $x-\chi_i+\chi_j\in S$ and $y+\chi_i-\chi_j\in S$.
\end{description}
It is known that M-convex sets are also characterized in terms of the following (seemingly weaker but actually equivalent) exchange property~\cite{Murota1998}:
\begin{description}
\item[(B-EXC${}_+$)]
For any $x,y\in S$ and for any $i\in\supp^+(x-y)$, there exists some $j\in\supp^-(x-y)$ such that $y+\chi_i-\chi_j\in S$.
\end{description}

An M-convex set $S$ is said to be \emph{matroidal M-convex} if 
$|x_e-y_e|\le 1$ for any $x,y\in S$ and any $e\in E$.%
\footnote{In other words, an M-convex set is matroidal if it is obtained from some matroid on $E$ by translating 
the characteristic vectors of the bases by the same integral vector.}
Lemma~\ref{lem:transfer} implies that the set of size vectors of the clean allocations and the clean utilitarian optimal allocations satisfy (B-EXC${}_+$).
\begin{lemma}\label{lem:M}
The following sets are M-convex:
\begin{align}
&S_1=\bigl\{(|A_0|,|A_1|,\dots,|A_n|)\in\mathbb{Z}^{N\cup\{0\}}\mid \text{$A$ is clean allocation}\bigr\} \quad\text{and} \label{eq:M0}\\
&S_2=\bigl\{(|A_1|,\dots,|A_n|)\in\mathbb{Z}^{N}\mid \text{$A$ is a clean utilitarian optimal allocation}\bigr\}.\label{eq:M1}
\end{align}
\end{lemma}
Note that, for each of the above M-convex sets $S_i$ for $i=1,2$,
the following problems are solvable in polynomial time via matroid intersection~\cite{Edmonds1970}:
\begin{description}
\item[(Initialization)] computing an element of $S_i$, and 
\item[(Membership)] deciding whether a given size vector is in $S_i$.
\end{description}
Also, for a given vector $x$ in $S_1$ or $S_2$, there is a polynomial time algorithm that finds an allocation whose size vector is equal to $x$; indeed, we can find such an allocation by computing a clean utilitarian optimal allocation for the profile $P'=(v_1',\dots,v_n')$ such that $v_i'(X)=\min\{v_i(X),\,x_i\}$ for each $i\in N$ and $X\subseteq N$.

Frank and Murota~\cite[Theorem 5.7]{FM19} proved that the set of Lorenz dominating elements\footnotemark of an M-convex set is a matroidal M-convex set. 
\footnotetext{For a given set of vectors, Lorentz dominating element is an element such that the smallest entry is as large as possible; within this, the next smallest entry is as large as possible; and so on.}
Furthermore, they showed that, in the matroidal M-convex set, a Lorenz dominating element that minimizes a linear function can be found in polynomial time if (Initialization) and (Membership) for the M-convex set can be solved in polynomial time.
By combining this with the fact that $S_2$ is a M-convex set, we obtain the following lemma.
\begin{lemma}\label{lem:Mdmin}
The set of size vectors corresponding to clean Lorenz dominating allocations $S^*\coloneqq\{\sv(A)\mid A\in\cLD\}$ is a matroidal M-convex set. 
Additionally, for a given weight $w\in\mathbb{R}^N$ a minimum-weight clean Lorenz dominating allocation $\argmin_{s\in S^*}\sum_{i\in N}w_is_i$ can be found in polynomial time.
\end{lemma}

Since $S^*$ is a matroidal M-convex set, the difference between values of the best and the worst clean Lorenz dominating allocations for each agent is at most one.
\begin{proposition}\label{prop:01}
\begin{align}
\max_{B\in\cLD}|B_i|-\min_{C\in\cLD}|C_i|\in\{0,1\} \quad \text{for any $i\in N$}. 
\end{align}
\end{proposition}

Combining this with the fact that $\sv^\uparrow(A)$ is unique across all $A\in\cLD$, we obtain the following essential property of the SE mechanism.

\begin{proposition}\label{prop:SE_uu}
The utility of each agent $i$ does not change according to the choice of allocation at Step~$1$ in the SE mechanism.
\end{proposition}
\begin{proof}
Let $A$ be the allocation chosen at Step 1. Then, for agent $i$ with $|A_i|<\max_{j \in N}|A_j|$, $i$ gets the final utility $\min_{B \in \cLD}|B_i| +1$, which does not depend on the choice of $A$ (note that if the difference is $1$, $\max_{B \in \cLD}|B_i|= \min_{B \in \cLD}|B_i| +1$). For agent $i$ with $|A_i|=\max_{j \in N}|A_j|$, $i$ gets no subsidy and the final utility is $|A_i|=\max_{j \in N}|A_j|=\max_{j \in N}|B_j|$ for any $B$ in $\cLD$ (because the size vectors are identical across all allocations in $\cLD$); again, this final utility does not depend on the choice of $A$.
\end{proof}

Note that $\min_{B\in\cLD}|B_i|$ can be computed in polynomial time for each $i$ using Lemma~\ref{lem:Mdmin}, e.g., by setting the weight $w$ as $w_i=0$ and $w_j=1$ for all $j\in N\setminus\{i\}$.
Hence, the outcome of the SE mechanism can be computed in polynomial time.
\begin{lemma}\label{lem:SE_poly}
The SE mechanism is polynomial time implementable.
\end{lemma}


Now, we show three other important properties of $\cLD$. First, we show that the minimum valuation realized by each agent in clean Lorenz dominating allocations is monotone with respect to a restriction of their valuation function.
Note that this property can be also proved via the monotonicity of the PE mechanism~\cite[Lemma~21]{babaioff2020fair} with respect to the priority ordering for which $i$ has the lowest priority. 
\begin{lemma}\label{lem:min_monotone}
For any $i\in N$ and sets $X\subseteq Y\subseteq M$, 
\begin{align}
\min_{B\in\cLD[P]}|B_i|&\le\min_{B'\in\cLD[P']}|B'_i|.\label{eq:monotone1}
\end{align}
where $P=(v_1,\dots,\restr{v_i}{X},\dots,v_n)$ and $P'=(v_1,\dots,\restr{v_i}{Y},\dots,v_n)$.
\end{lemma}
\begin{proof}
It is sufficient to prove the case when $|Y|=|X|+1$ (since we can apply induction).
Let $a$ be the item that is in $Y$ but not in $X$ (i.e., $Y=X\cup\{a\}$).

Let $A\in\argmin_{B\in\cLD[P]}|B_i|$ and $A'\in\argmin_{B'\in\cLD[P']}|B'_i|$ such that $\sum_{j\in N}j\cdot|B_j|$ and $\sum_{j\in N}j\cdot|B'_j|$ are minimized, respectively.
Suppose to the contrary, we assume that $|A_i|>|A'_i|$.
Let $R$ be the set of agents $s$ such that $|A_s|>|A'_s|$ and subject to that $|A_s|$ is minimized. If $R=\{i\}$, then let $s=i$, and otherwise, let $s$ be the agent with smallest index in $R\setminus\{i\}$.
As $A'$ allocates at least as many items as $A$, there exists an agent $j$ such that $|A_j|<|A'_j|$.
Let $t$ be the agent with $|A_t|<|A'_t|$ such that $|A_t|$ is minimized, and if there are multiple such agents, choose the one with the smallest index.

Consider the case where $|A_s|<|A'_t|$, or $|A_s|=|A'_t|$ and $i\ne s<t$.
By the exchange property for $A$ and $A'$ with $|A_s|>|A'_s|$ (recall that \(\{(|B'_0|,|B'_1|,\dots,|B'_n|)\mid \text{$B'$ is clean under $P'$}\}\) is M-convex by \eqref{eq:M0}), there exists a clean allocation $C$ under $P'$ such that 
$\sv(C)=\sv(A')+\chi_s-\chi_k$ for some $k\in N$ with $|A'_k|>|A_k|$.
Note that $|A'_s|<|A_s|\le|A'_t|\le|A'_k|$.
If $|A'_s|\le |A'_k|-2$, then $A'$ does not Lorenz dominate $C$, a contradiction.
Otherwise, i.e., $|A'_s|+1=|A_s|=|A'_t|=|A'_k|$, we have $C\in\cLD[P']$ and $|C_i|=|A'_i|$ by $s\ne i$.
By $s<t\le k$, we have $\sum_{j\in N}j\cdot |C_j|<\sum_{j\in N}j\cdot|A'_j|$, a contradiction. 

Finally, consider the other case, i.e., (i) $|A_s|>|A'_t|$, (ii) $|A_s|=|A'_t|$ and $i\ne s>t$, or (iii) $|A_s|=|A'_t|$ and $s=i$.
Let $A''=(A'_1,\dots,A'_i\setminus\{a\},\dots,A'_n)$.
By the exchange property for $A$ and $A''$ with $|A_t|<|A'_t|$ (recall that \(\{(|B_0|,|B_1|,\dots,|B_n|)\mid \text{$B$ is clean under $P$}\}\) is M-convex by \eqref{eq:M0}), there exists a clean allocation $D$ under $P$ such that 
$\sv(D)=\sv(A)+\chi_t-\chi_\ell$ for some $\ell\in N$ with $|A''_\ell|<|A_\ell|$ (here, $\ell\ne 0$ by $|A''_0|=|M|-\sum_{j\in N}|A''_j|\ge |M|-\sum_{j\in N}|A_j|=|A_0|$).
Note that $|A_t|<|A'_t|\le |A_s|\le |A_\ell|$.
If $|A_t|\le |A_\ell|-2$, then $A$ does not Lorenz dominate $D$, a contradiction.
Otherwise, i.e., $|A_t|+1=|A'_t|=|A_s|=|A_\ell|$, we have $D\in\cLD[P]$.
If $\ell\ne i$, we have $\sum_{j\in N}j\cdot |D_j|<\sum_{j\in N}j\cdot|A'_j|$ by $t<s\le \ell$, a contradiction. 
If $\ell=i$, we have $|D_i|<|A_i|$, which contradicts the assumption that $A\in\argmin_{B\in\cLD[P]}|B_i|$.
\end{proof}



We next show that, if an agent receives a bundle $A_i$ and changes her report to the restriction $\restr{v_i}{X}$ for $X \subseteq A_i$,
then she would be allocated the set $X$ at some allocation in $\cLD$.
Again, this property can be proved by the strong faithfulness of the PE mechanism \cite[Lemma~22]{babaioff2020fair} with respect to appropriate priority orders.
\begin{lemma}\label{lem:sfaithful0}
Let $P=(v_1,\dots,v_n)$, $i\in N$, $A\in\cLD[P]$, $X\subseteq A_i$, and $P'=(v_1,\dots,\restr{v_i}{X},\dots,v_n)$.
Then, we have
\begin{itemize}
    \item $B_i=X$ for some $B\in\cLD[P']$, and
    \item if $|A_i|=\min_{A'\in\cLD[P]}|A'_i|$, then $B_i=X$ for any $B\in\cLD[P']$.
\end{itemize}
\end{lemma}
\begin{proof}
Let $P''=(v_1,\dots,\restr{v_i}{A_i},\dots,v_n)$.
Note that $A\in\cLD[P'']\subseteq\cLD[P]$ since $A$ is clean under $P''$.
In addition, $B''_i=A_i$ for any $B''\in\cLD[P'']$ if $|A_i|=\min_{A'\in\cLD[P]}|A'_i|$ (since otherwise $B''_i\subsetneq A_i$ for some $B''\in\cLD[P'']\subseteq\cLD[P]$, and hence $|B''_i|<|A_i|=\min_{A'\in\cLD[P]}|A'_i|\le \min_{A''\in\cLD[P'']}|A''_i|$, a contradiction).

To prove the proposition, we show that 
\begin{align}
\text{$C_i=Y$ for all $C\in\cLD[Q]$ if $C'_i=Z$ for some $C'\in\cLD[Q']$}
\end{align}
for any $Y\subsetneq Z\subseteq M$ where $Q=(v_1,\dots,\restr{v_i}{Y},\dots,v_n)$ and $Q'=(v_1,\dots,\restr{v_i}{Z},\dots,v_n)$.
It is sufficient to prove the case when $|Z|=|Y|+1$ (since we can apply induction).
Let $a$ be the item that is in $Z$ but not in $Y$ (i.e., $X=Y\cup\{a\}$).

Assume towards a contradiction that $C_i\subsetneq Y$ for some $C\in\cLD[Q]$ but $C'_i=Z$ for some $C'\in\cLD[Q']$.
Let $C\in\cLD[Q]$ be an allocation with $|C_i|=\min_{A'\in\cLD[Q]}|A'_i|$ such that $\sum_{j\in N}j\cdot|C_j|$ is minimized.
Similarly, let $C'\in\cLD[Q']$ be an allocation with $C'_i=Z$ such that $\sum_{j\in N}j\cdot|C'_j|$ is minimized.
In addition, let $C''=(C'_1,\dots,C'_i\setminus\{a\},\dots,C'_n)$.
Note that $C''_i=Y$ and $C''$ is clean under $Q$.

Let $R$ be the set of agents $s$ such that $|C_s|<|C''_s|$ and subject to that $|C''_s|$ is minimized. If $R=\{i\}$, then let $s=i$, and otherwise, let $s$ be the agent with smallest index in $R\setminus\{i\}$.
As $C$ allocates at least as many items as $C''$, there exists an agent $j$ such that $|C_j|>|C''_j|$.
Let $t$ be the agent with $|C_t|>|C''_t|$ such that $|C_t|$ is minimized, and if there are multiple such agents, choose the one with the smallest index.

Consider the case where $|C''_s|<|C_t|$, or $|C''_s|=|C_t|$ and $i\ne s<t$.
By the exchange property for $C$ and $C''$ with $|C_s|<|C''_s|$, (recall that \(\{(|A'_0|,|A'_1|,\dots,|A'_n|)\mid \text{$A'$ is clean under $Q$}\}\) is M-convex by \eqref{eq:M0}), 
there exists a clean allocation $D$ under $Q$ such that 
$\sv(D)=\sv(C)+\chi_s-\chi_k$ for some $k\in N$ with $|C_k|>|C''_k|$.
Note that $|C_s|<|C''_s|\le|C_t|\le|C_k|$.
If $|C_s|\le |C_k|-2$, then $C$ does not Lorenz dominate $D$, a contradiction.
Otherwise, i.e., $|C_s|+1=|C''_s|=|C_t|=|C_k|$, we have $D\in\cLD[Q]$ and $|D_i|=|C_i|$ by $s\ne i$.
By $s<t\le k$, we have $\sum_{j\in N}j\cdot |D_j|<\sum_{j\in N}j\cdot|C_j|$, a contradiction. 

Finally, consider the other case, i.e., (i) $|C''_s|>|C_t|$, (ii) $|C''_s|=|C_t|$ and $i\ne s>t$, or (iii) $|C''_s|=|C_t|$ and $s=i$.
By the exchange property for $C$ and $C'$ with $|C_t|>(|C''_t|=)|C'_t|$,
(recall that \(\{(|A'_0|,|A'_1|,\dots,|A'_n|)\mid \text{$A'$ is clean under $Q'$}\}\) is M-convex by \eqref{eq:M0}), 
there exists a clean allocation $D'$ under $Q'$ such that 
$\sv(D')=\sv(C')+\chi_t-\chi_\ell$ for some $\ell\in N$ with $|C'_\ell|>|C_\ell|$.
Note that $|C'_t|<|C_t|\le|C''_s|\le|C''_\ell|\le |C'_\ell|$.
If $|C'_t|\le |C'_\ell|-2$, then $C'$ does not Lorenz dominate $D'$, a contradiction.
Otherwise, i.e., $|C'_t|+1=|C_t|=|C''_s|=|C''_\ell|=|C'_\ell|$, we have $D'\in\cLD[Q']$ and $\ell\ne i$.
Thus, we have $\sum_{j\in N}j\cdot |D'_j|<\sum_{j\in N}j\cdot|C'_j|$ by $t<s\le\ell$, a contradiction. 

\end{proof}


Finally, we analyze the effect of restriction upon the number of items allocated to each agent.
\begin{lemma}\label{lem:maximp}
Fix agent $i$. Let $P=(v_1,\dots,\restr{v_i}{X},\dots,v_n)$ and $P'=(v_1,\dots,\restr{v_i}{Y},\dots,v_n)$ for subsets $X\subseteq Y\subseteq M$.
Suppose that for some $A'\in\cLD[P']$, $A'_i=Y$, and $i$'s bundle has a strictly smaller size than the largest bundle, i.e., $|A'_i|<\max_{j\in N}|A'_j|$. Then,  $|A_i|<\max_{j\in N}|A_j|$ for any $A\in\cLD[P]$.
\end{lemma}
\begin{proof}
It is sufficient to prove the case when $|Y|=|X|+1$ because we can apply induction if $|Y|>|X|$.
Let $a$ be the item that is in $Y$ but not in $X$, i.e., $Y=X\cup\{a\}$.

Suppose that $|A'_i|<\max_{j\in N}|A'_j|$ and $A'_i=Y$ for some $A'\in\cLD[P']$.
Consider any allocation $A\in\cLD[P]$ and let $A''=(A'_1,\dots,A'_i\setminus\{a\},\dots,A'_n)$.
As $A$ is clean under $P'$ and $A''$ is clean under $P$, we have
\begin{align}
\sum_{j=1}^k\sv^\uparrow(A'')_j\le
\sum_{j=1}^k\sv^\uparrow(A)_j\le
\sum_{j=1}^k\sv^\uparrow(A')_j\label{eq:maximp}
\end{align}
for any $k\in\{1,\dots,n\}$.
Let $j^*$ be the index such that $\sv^\uparrow(A')_{j^*-1}<|X|$ and $\sv^\uparrow(A')_{j^*}\ge |X|$.
Note that $|A'_i|$ is placed on or after $j^*$ because $|A'_i|=|X|+1$.
Then, for any $j<j^*$, we see that $\sv^\uparrow(A')_j=\sv^\uparrow(A'')_j$, and hence $\sv^\uparrow(A)_j=\sv^\uparrow(A')_j$.
Here, $\sum_{j\in N}|A_j|\ge \sum_{j\in N}|A''_j|=\sum_{j\in N}|A'_j|-1$ by \eqref{eq:maximp} with $k=n$.
As $|X|<|Y|=|A'_i|<\max_{j\in N}|A'_j|$, we have $|X|+2\le\max_{j\in N}|A'_j|$. This together with $\sv^\uparrow(A')_j\ge |X|$ for each $j\ge j^*$, we have
\begin{align}
\sum_{j=j^*}^n \sv^\uparrow(A)_j
\ge \sum_{j=j^*}^n \sv^\uparrow(A')_j-1
\ge \sum_{j=j^*}^n |X|+1.
\end{align}
Thus, $\max_{j\in N}|A_j|\ge |X|+1$.
As $A_i\subseteq X$, we obtain $|A_i|\le |X|<\max_{j\in N}|A_j|$.
\end{proof}

\subsection{Envy-freeness of the SE mechanism}
We are now ready to show that the SE mechanism is envy-free. 
\begin{lemma}\label{lem:SE_EF}
The SE mechanism is envy-free.
\end{lemma}
\begin{proof}
Let $(A,p)$ be the clean allocation with a subsidy returned by the SE mechanism.
To obtain a contradiction, suppose that $i$ envies $j$, i.e., $v_i(A_i)+p_i<v_i(A_j)+p_j$.
We separately consider three cases: $v_i(A_i)>v_i(A_j)$, $v_i(A_i)<v_i(A_j)$, and $v_i(A_i)=v_i(A_j)$.

\smallskip
\noindent\textbf{Case 1.} Suppose that $v_i(A_i)>v_i(A_j)$.
This case is impossible since $v_i(A_i)+p_i<v_i(A_j)+p_j$ and $p_i,p_j\in\{0,1\}$.

\smallskip
\noindent\textbf{Case 2.} Suppose that $v_i(A_i)<v_i(A_j)$.
By the matroid augmentation property, there exists an item $e\in A_j$ such that $v_i(A_i\cup\{e\})=v_i(A_i)+1$.
Let $B$ be the allocation that is obtained from $A$ by moving item $e$ from $j$ to $i$.
As $|A_i|<|A_j|$ and $A$ Lorenz dominates $B$, we have that $|B_i|=|A_i|+1=|A_j|=|B_j|+1$. 
Hence, $B$ is also a clean Lorenz dominating allocation.
Thus, $\max_{C\in\cLD}|C_i|=|A_i|+1$ and $\min_{C\in\cLD}|C_j|=|A_j|-1$, which implies $p_i=1$ and $p_j=0$ by Proposition~\ref{prop:01}.
This contradicts the assumption that $i$ envies $j$ because $v_i(A_i)+p_i=|A_i|+1=|A_j|=v_i(A_j)+p_j$. 

\smallskip
\noindent\textbf{Case 3.} Suppose that $v_i(A_i)=v_i(A_j)$.
Note that $|A_i|=v_i(A_i)=v_i(A_j) \leq |A_j|$. 
As $v_i(A_i)+p_i<v_i(A_j)+p_j$, the subsidies must be $p_i=0$ and $p_j=1$.
Then $|A_j|=\min_{A'\in\cLD}|A'_j|<\max_{k\in N}|A_k|$. 
We observe that $\min_{A'\in\cLD}|A'_i|=|A_i|-1$, because otherwise $\min_{A'\in\cLD}|A'_i|=|A_i|=\max_{k\in N}|A_k| > |A_j|$, and hence $v_i(A_i)=|A_i|>|A_j| \geq v_i(A_j)$, which is a contradiction.
As $\cLD$ is an M-convex set, there is a clean Lorenz dominating allocation $B$ such that $\sv(B)=\sv(A)-\chi_i+\chi_k$ for some $k\in N$.
As $A$ and $B$ are both in $\cLD$ and hence $\sv^\uparrow(A)=\sv^\uparrow(B)$, we have that $|B_i|+1=|A_i|=|B_k|=|A_k|+1$.
Note that $k\ne j$ because $|A_i|\le |A_j|$ by $v_i(A_i)=v_i(A_j)$.

By applying Lemma~\ref{lem:transfer} to $B$ and $A$ (note that the roles are interchanged), 
we obtain a sequence of clean allocations $C^0,C^1,\dots,C^r$ with $k^0,k^1,\dots,k^r$ and $e^1,\dots,e^r$ where $C^0=A$, $k^0=k$, $k^r=i$, and $\sv(C^r)=\sv(C^0)+\chi_{k^0}-\chi_{k^r}=\sv(B)$.
If $k^t=j$ for some $t$, then $\sv(C^t)=\sv(A)+\chi_k-\chi_j$ and $|A_k|+1=|A_i|\leq|A_j|$, and hence $C^t$ is a clean Lorenz dominating allocation with $|C^t_j|<|A_j|$.
This implies $p_j=0$, which is a contradiction.
Otherwise, i.e., $k^t\ne j$ for all $t$, we have $C^r_j=A_j$. Then, there exists an element $e\in C^r_j$ such that $v_i(C^r_i\cup\{e\})=|A_i|$ by $v_i(C^r_i)=|A_i|-1<|A_i|=v_i(C^r_j)$ and the matroid augmentation property. Thus, the allocation that is obtained from $C^r$ by transferring $e$ from $j$ to $i$ is a clean Lorenz dominating allocation. This also implies that $p_j=0$, which is a contradiction.
\end{proof}

\subsection{Truthfulness of the SE mechanism}
Finally, we prove that the SE mechanism is truthful.
In a setting without money, Babaioff et al.~\cite{babaioff2020fair} proved that a mechanism is truthful if it satisfies \emph{monotonicity} and \emph{strong faithfulness}. We introduce two similar properties that can be applied to a setting with subsidies.

\begin{definition}[\smonotone]
We say that a mechanism is \emph{\smonotone} if the utility of an agent is monotone with respect to the restriction, i.e.,
\[v_i(A_{i})+p_{i}\le v_i(A'_{i})+p'_{i}\] 
for any valuation function $(v_1,\dots,v_n)$, agent $i\in N$, and subsets  $X\subseteq Y\subseteq M$, where $(A,p)$ and $(A',p')$ are the allocations with subsidies returned by the mechanism when agents report $P=(v_1,\ldots,\restr{v_i}{X},\ldots,v_n)$ and $P'=(v_1,\ldots,\restr{v_i}{Y},\ldots,v_n)$, respectively.
\end{definition}

\begin{definition}[\sfaithful]
We say that a mechanism is \emph{\sfaithful} if 
\[v_i(X)+p_i\le v_i(A'_i)+p'_i\] 
for any valuation function $(v_1,\dots,v_n)$, agent $i\in N$, and subset $X\subseteq A_i$, where $(A,p)$ and $(A',p')$ are the allocations with subsidies returned by the mechanism when agents report $P=(v_1,\ldots,v_i,\ldots,v_n)$ and $P'=(v_1,\ldots,\restr{v_i}{X},\ldots,v_n)$, respectively.
\end{definition}

\begin{lemma}\label{lem:ss}
A mechanism is truthful if it is \smonotone and \sfaithful.
\end{lemma}
\begin{proof}
Let $(A,p)$ be the allocation with a subsidy returned by the mechanism when agents report $(v_1,\ldots,v_{i},\ldots,v_n)$ truthfully; let $v'_i$ be the matroidal valuation function for $i$ such that $v'_i \neq v_i$; let $(A',p')$ be the allocation with a subsidy returned by the mechanism when agents report $(v_1,\ldots,v'_{i},\ldots,v_n)$. We will show that agent $i$ will not benefit from misreporting $v'_i$, i.e.,
\begin{align}\label{eq:truthful}
v_i(A'_i) +p'_i \leq v_i(A_i) + p_i.
\end{align}

To this end, let $X$ be a minimum subset of $A'_i$ such that $v_i(X)=v_i(A'_i)$; equivalently, $X$ is a maximum-size independent set contained in $A'_i$ under the valuation $v_i$. 
Let $(A'',p'')$ be the allocation with the subsidy returned by the mechanism when agents report $(v_1,\ldots,\restr{v'_i}{X},\ldots,v_n)$. 
By the \sfaithful~property, 
\begin{align}\label{eq:faithful}
v_i(A'_i)+p'_i=|X|+p'_i=v'_i(X)+p'_i \leq v'_i(A''_i) + p''_i=|A''_i|+p''_i.
\end{align}
Further, since $\restr{v'_{i}}{X}= \restr{v_{i}}{X}$, agent $i$ obtains $A''_i$ together with $p''_i$ when $i$ reports $\restr{v_i}{X}$. 
Now, by the \smonotone~property, 
\begin{align}
|A''_i|+p''_i = v_i(A''_i) + p''_i \leq  v_i(A_i) +p_i,  
\end{align}
which, together with \eqref{eq:faithful}, implies inequality \eqref{eq:truthful}. 
\end{proof}

Below, we prove that the SE mechanism is \smonotone and \sfaithful.
\begin{lemma}\label{lem:smonotone}
The SE mechanism is \smonotone. 
\end{lemma}
\begin{proof}
Consider any agent $i$ and sets $X \subseteq Y \subseteq M$. 
It is sufficient to prove the case when $|T|=|S|+1$ since we can apply induction.
Let $a$ be the item that is in $Y$ but not in $X$ (i.e., $Y=X\cup\{a\}$); let $(A,p)$ and $(A',p')$ be the allocations with the subsidy returned by the mechanism when agents report $P=(v_1,\ldots,\restr{v_i}{X},\ldots,v_n)$ and $P'=(v_1,\ldots,\restr{v_i}{Y},\ldots,v_n)$, respectively. 
Without loss of generality, we assume that, $|A_i|=\min_{B\in\cLD[P]}|B_i|$ and $|A'_i|=\min_{B'\in\cLD[P']}|B'_i|$ (recall that the utility of every agent does not change with the choice of clean Lorenz dominating allocation).
By Lemma~\ref{lem:min_monotone}, we have $|A_i|\le |A'_i|$.

If $|A_i|<|A'_i|$ or $(p_i,p'_i)\ne (1,0)$, then
\begin{align}
v_i(A_i)+p_i=|A_i|+p_i\le |A'_i|+p'_i=v_i(A'_i)+p'_i.
\end{align}
Hence, we only need to prove that $|A_i|=|A'_i|$ and $(p_i,p'_i)=(1,0)$ cannot be satisfied simultaneously.
To the contrary, suppose that $|A_i|=|A'_i|$ and $(p_i,p'_i)=(1,0)$.
Let $\gamma=|A_i|$. 

If $a$ is not in $A'_i$, the allocation $A'$ must be in $\cLD[P]$ and hence $\sv^\uparrow(A')=\sv^\uparrow(A)$.
By $(p_i,p'_i)=(1,0)$, we have $\gamma<\max_{j\in N}|A_j|=\max_{j\in N}|A'_j|=\gamma$, a contradiction.
Thus, we assume that $a\in A'_i$.

Let $A''=(A'_1,\dots,A'_i\setminus\{a\},\dots,A'_n)$ and $s=|\{j\in N\mid |A'_j|=\gamma\}|$.
Then, $A$ Lorenz dominates $A''$ since $A''$ is clean under $P$.
Also, $A'$ Lorenz dominates $A$ since $A$ is clean under $P'$.
Hence, we obtain
\begin{align}
\sum_{j=1}^{k}\sv^\uparrow(A'')_j\le \sum_{j=1}^{k}\sv^\uparrow(A)_j\le\sum_{j=1}^{k}\sv^\uparrow(A')_j
\end{align}
for any $k=1,\dots,n$.
For $j=1,2,\dots,n-s$, we have $\sv^\uparrow(A')_j=\sv^\uparrow(A'')_j$, and hence $\sv^\uparrow(A)_j=\sv^\uparrow(A')_j=\sv^\uparrow(A'')_j$.
Note that $\sv^\uparrow(A)_n>\gamma$ by $p_i=1$.
As $A'$ allocates at least as many items as $A$, we have $\sv^\uparrow(A)=\sv^\uparrow(A')-\chi_{n-s+1}+\chi_n$ (see Figure~\ref{fig:smonotone}).
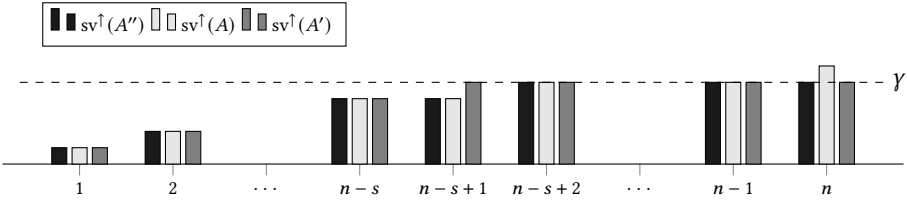
\begin{figure}
    \centering
    \begin{tikzpicture}[scale=1,a/.style={draw=black,fill=black!10},b/.style={draw=black,fill=black!50},c/.style={draw=black,fill=black!90}]
    \begin{axis}[
    ybar, 
    bar width=0.2cm,
    enlarge x limits=0.05,
    width=0.9\textwidth,
    height=3cm,
    major grid style={draw=black},
    enlarge y limits={value=.1,upper},
    ymin=0, ymax=6,
    axis x line*=bottom,
    axis y line*=right,
    hide y axis,axis line style={shorten >=-15pt, shorten <=-15pt},
    xticklabels={$1\strut$,$2\strut$,$\cdots\strut$,$n-s\strut$,$n-s+1\strut$, $n-s+2\strut$, $\cdots\strut$, $n-1\strut$, $n\strut$},
    xtick={1,2,3,4,5,6,7,8,9},
    xticklabel style = {font=\scriptsize,yshift=0.5ex},
    legend style={at={(0,1.5)},anchor=north west,legend columns=-1,font=\scriptsize},
    ]
    
    \addplot [c] coordinates {
       (1,1) (2,2) (4,4) (5,4) (6,5) (8,5) (9,5)};
    \addlegendentry{$\sv^\uparrow(A'')$};
    \addplot [a]  coordinates {
       (1,1) (2,2) (4,4) (5,4) (6,5) (8,5) (9,6) } coordinate (g);
    \addlegendentry{$\sv^\uparrow(A)$};
    \addplot [b] coordinates {
       (1,1) (2,2) (4,4) (5,5) (6,5) (8,5) (9,5)};
    \addlegendentry{$\sv^\uparrow(A')$};
    \end{axis}
    \draw[dashed] (${5/6}*(g)+(-9,0)$)--++(11.5,0) node[right] {$\gamma$};
\end{tikzpicture}
    \caption{Example of $\sv^\uparrow(A'')$, $\sv^\uparrow(A)$, $\sv^\uparrow(A')$. The $i$th block corresponds to the $i$th entries in the vectors, and the height of each bar represents the value.}
    \label{fig:smonotone}
\end{figure}
Thus, $A$ and $A'$ allocate the same number of items, i.e., $\sum_{j\in N}|A_j|=\sum_{j\in N}|A'_j|$.
By the exchange property for $A''$ and $A$ with $|A''_0|>|A_0|$ (recall that \(\{(|B_0|,|B_1|,\dots,|B_n|)\mid \text{$B$ is clean under $P$}\}\) is M-convex by \eqref{eq:M0}), there exists a clean allocation $C$ under $P$ such that $\sv(C)=\sv(A'')+\chi_\ell$ for some $\ell\in N$. 
As $A$ Lorenz dominates $C$ (and hence $\sum_{j=1}^{n-1}\sv^\uparrow(A)_j\ge\sum_{j=1}^{n-1}\sv^\uparrow(C)_j$), $\sum_{j\in N}|A_j|=\sum_{j\in N}|C_j|$, and $\sv^\uparrow(A)_n>\gamma$, we have $|C_\ell|=\gamma+1$ and $\ell\ne i$ (recall that $|A''_i|=\gamma-1$ and $|A''_j|\le\gamma$ for all $j\in N$).
Then, $\sv^\uparrow(C)=\sv^\uparrow(A)$ and hence $C\in\cLD[P]$.
This implies that $|A_i|=\gamma>\gamma-1=|A''_i|=|C_i|\ge \min_{B\in\cLD[P]}|B_i|$, which contradicts the assumption that $|A_i|=\min_{B\in\cLD[P]}|B_i|$.
\end{proof}

\begin{lemma}\label{lem:sfaithful}
The SE mechanism is \sfaithful.
\end{lemma}
\begin{proof}
Consider any agent $i$. Let $(A,p)$ be the allocation with a subsidy returned by the SE mechanism when agents report $P=(v_1,\ldots,v_{i},\ldots,v_n)$ and fix any $X \subseteq A_i$.
Let $(A',p')$ and $(A'',p'')$ be the allocations with the subsidies returned by the SE mechanism when agents report $P'=(v_1,\ldots,\restr{v_i}{A_i},\ldots,v_n)$ and $P''=(v_1,\ldots,\restr{v_i}{X},\ldots,v_n)$, respectively.
If $p_i=0$, we have $v_i(X)+p_i=|X|\le \max_{B''\in\cLD[P'']}|B''_i|\le v_i(A''_i)+p''_i$ by Lemma~\ref{lem:sfaithful0}.

In what follows, we assume $p_i=1$, i.e., $|A_i|=\min_{B\in\cLD[P]}|B_i|<\max_{j\in N}|A_j|$.
By Lemma~\ref{lem:sfaithful0}, $B'_i=A_i$ for any $B'\in\cLD[P']$; thus, $A'_i=A_i$ and $|A'_i|=|A_i|<\max_{j\in N}|A_j|=\max_{j\in N}|A'_j|$ (the last equality holds since $\sv^\uparrow(A)=\sv^\uparrow(A')$ by $A,A'\in\cLD[P']$).
By Lemma~\ref{lem:sfaithful0}, $B''_i=X$ for any $B''\in\cLD[P'']$, and in particular $A''_i=X$.
Also, by Lemma~\ref{lem:maximp}, $|A''_i|<\max_{j\in N}|A''_j|$.
Therefore, $p'_i=1$ and hence $v_i(X)+p_i=|X|+1=v_i(A''_i)+p''_i$.
\end{proof}

By combining Lemmas~\ref{lem:smonotone} and~\ref{lem:sfaithful} and using Lemma~\ref{lem:ss}, we obtain the desired truthfulness.
\begin{lemma}\label{lem:SE_truthful}
The SE mechanism is truthful.
\end{lemma}

\subsection{Without the free-disposal assumption}\label{sec:withoutfreedisposal}
In Theorem~\ref{thm:matroidrank:truthful}, 
we presented the so-called SE mechanism, which simultaneously attains truthfulness, utilitarian optimality, and
envy-freeness with each agent receiving a subsidy of $0$ or $1$.
In the mechanism's output, however, the allocation may not be complete (i.e., some items may be left unallocated).
In some situations, this disposal of items is not allowed. 
For example, when we consider a shift scheduling at a call center or a production factory, we must allocate all shifts to employees in order not to stop the operation, even if no one finds value in that time slot. 

It is ideal if there is a mechanism that outputs a complete allocation while attaining the nice properties of the SE mechanism (i.e., truthfulness, utilitarian optimality, and envy-freeness with each agent receiving a subsidy of at most $1$).
However, as shown below, the amount of subsidy needed to pay can be proportional to the number of items if we aim to attain truthfulness, envy-freeness, and completeness, while using Lorenz dominating allocations.
\begin{theorem}\label{thm:impossibility2}
If a truthful mechanism is envy-free, and returns a complete Lorenz dominating allocation, it requires a subsidy of $\Omega(m)$, even when there are two agents with binary additive valuations.
\end{theorem}
\begin{proof}
The following proof is inspired by the proof of \cite[Theorem 5]{halpern2020fair}, which shows the nonexistence of a mechanism without money that satisfies truthfulness, EFX, and completeness, while returning a Lorenz dominating allocation. Fix a positive integer $k$; let $N=\{1,2\}$ and $M=\{e_1,\dots,e_{6k}\}$.
Consider two profiles, $P_1$ and $P_2$.
In $P_1$, both agents report that they want the $2k$ items $\{e_1,\dots,e_{2k}\}$.
In this case, each agent receives exactly $k$ items by Lorenz dominance; additionally, one agent receives at least half of the items in $\{e_{2k+1},\dots,e_{6k}\}$.
Without loss of generality, agent $1$ gets a set of items including $\{e_1,\dots,e_k,e_{2k+1},\dots,e_{4k}\}$.
In $P_2$, agent $1$ reports that she wants $\{e_1,\dots,e_{4k}\}$ and agent $2$ reports that she wants $\{e_1,\dots,e_{2k}\}$.
In this case, the items $e_{1},\dots,e_{2k}$ are allocated to agent $2$ and the items $e_{2k+1},\dots,e_{4k}$ are allocated to agent $1$.
If $P_2$ is the true valuation profile, agent $1$ has an incentive to report that she wants $\{e_1,\dots,e_{2k}\}$ unless she gets at least $k~(=\Omega(m))$ subsidy because she receives $2k$ valuable items in $P_2$ but $3k$ in $P_1$.
\end{proof}

Here, we provide an algorithm that returns a Lorenz dominating allocation and simultaneously attains completeness and envy-freeness with each agent receiving a subsidy of at most $1$ while tolerating a violation of truthfulness. 
Note that an allocation that is both complete and envy-freeable with a subsidy of at most $1$ (where $1$ is the maximum marginal value) for each agent has been shown to exist only for additive valuations \cite{halpern2020fair, Brustle2020}. 
The following theorem guarantees the existence of such an allocation for matroidal valuations, which are non-additive.

\begin{theorem}\label{thm:matroidrank:complete}
For matroidal valuations, 
there is a polynomial-time algorithm for computing an allocation with a subsidy that is complete, utilitarian optimal, and envy-free, with each agent receiving a subsidy of $0$ or $1$ and the total subsidy being at most $n-1$.
\end{theorem}
We construct the allocation required in the theorem by extending an arbitrary clean Lorenz dominating allocation $A=(A_1, A_2,\dots,A_n)$; that is, we initialize $A$ to be the one computed in Step 1 of the SE mechanism.
By Theorem~\ref{thm:matroidrank:truthful}, $A$ then maximizes the utilitarian social welfare $\sum_{i\in N}v_i(A_i)$ and is envy-freeable with a subsidy of at most 1 for each agent.
Therefore, we can obtain a desired allocation if we can allocate items in $M\setminus \bigcup_{i\in N}A_i$
while preserving the utilitarian optimality and the bound 1 of the subsidy for each agent.
Note that, for binary additive valuations, this task is trivial because an item unallocated in $A$ 
has a value of $0$ for all agents by the utilitarian optimality; hence allocating it to any agent does not cause envy. 
However, a similar argument does not apply to matroidal valuations, as shown by the following example.

\begin{example}
Let $N=\{1,2,3\}$ and $M=\{e_1,e_2,e_3,e_4,e_5\}$ and define the matroidal valuations $v_1,v_2,v_3$ by
$v_1(X)=|X\cap\{e_1,e_2\}|$, $v_2(X)=|X\cap \{e_1,e_2,e_3\}|$, and $v_3(X)=|X\cap \{e_1,e_2,e_3\}|+\min\bigl\{1,\,|X\cap \{e_4,e_5\}|\bigr\}$.
Then $A=(A_1,A_2,A_3)=\bigl\{\{e_1,e_2\},\{e_3\},\{e_4\}\bigr\}$ is a clean Lorenz dominating allocation.
It is not difficult to see that we cannot increase the utility of any agent by allocating $e_5$, which is currently unallocated.
However, if we allocate $e_5$ to agent $2$, the amount $w(3,2)=v_3(A_2)-v_3(A_3)$ 
of envy agent $3$ has towards $2$ changes from $0$ to $1$. 
To eliminate envy for the resultant allocation 
$A'=(A'_1,A'_2,A'_3)=\bigl\{\{e_1,e_2\},\{e_3,e_5\},\{e_4\}\bigr\}$, we need to pay at least one dollar to agent 2 because her envy towards agent 1 is $v_2(A'_1)-v_2(A'_2)=1$.
Then $v_3(A'_2)+p_2\geq 3$ while $v_3(A'_3)=1$, and to eliminate the envy of agent $3$ towards agent 2, 
we must pay at least $2$ dollars to agent $3$.
\end{example}

Here, we present the {\em subsidized egalitarian with completion} (SEC) algorithm, which extends any clean Lorenz dominating allocation to a complete allocation while preserving the property that each agent requires at most 1 subsidy.
Recall that, as defined in Section~\ref{sec:model}, the envy graph $G_A$ for an allocation $A$ is 
a complete directed graph with a node set $N$ in which the arc weights represent the amounts of envies
with respect to $A$. Since matroidal valuations are integer-valued, each arc weight is an integer.
\begin{tcolorbox}[title=Subsidized Egalitarian with Completion, left=0mm]
\begin{enumerate}[label=\textbf{Step \arabic*.},leftmargin=*]
    \item Allocate items according to an arbitrarily chosen $A \in \cLD$.
    \item For each unallocated item $e\in M\setminus\bigcup_{i\in N}A_i$, do the following.
\begin{itemize}
\setlength{\leftskip}{3mm}
	\item[(a)] Take an agent $i$ arbitrarily.
	\item[(b)] Let $A^{i,e}\coloneqq(A_1,\dots,A_i\cup\{e\},\dots,A_n)$.
If $G_{A^{i,e}}$ has a positive-weight path ending at $i$, then take such a path 
$P_i$ arbitrarily and update $i$ by the initial agent of $P_i$ and go to (b). Otherwise, go to (c).
	\item[(c)] Update $A\gets A^{i,e}$ (i.e., $A_i\gets A_i\cup\{e\}$).
\end{itemize}
    \item Give $1$ subsidy to each agent $i\in N$ such that the envy graph $G_{A}$ has a path of weight $1$ starting at $i$.
\end{enumerate}
\end{tcolorbox}
From the description, even the finite termination of the SEC algorithm is not obvious.
We will show this later and here provide the conditions preserved in the algorithm.
\begin{lemma}\label{lem:complement2}
The following conditions hold throughout the SEC algorithm.
\begin{itemize}
	\item[\rm (i)] $A=(A_1,\dots,A_n)$ is utilitarian optimal.
	\item[\rm (ii)] $G_{A}$ has neither a path of weight more than $1$ nor a positive-weight cycle.
\end{itemize}
\end{lemma}
\begin{proof}
Just after Step 1, the conditions (i) and (ii) follow from Theorems~\ref{thm:envy-graph} and \ref{thm:matroidrank:truthful}.
Since each $A_i$ is monotone increasing in Step 2, the condition (i) is clearly preserved throughout the algorithm.
Then, $G_{A}$ has no positive-weight cycle, since otherwise we can increase the utilitarian social welfare 
by exchanging bundles along that cycle, in contradiction of (i).

We now show that, the nonexistence of a path of weight greater than $1$ is preserved whenever $A$ is updated.
Let $i\in N$ be an agent who receives an item $e$ in Step 2 (c).
By Step 2 (b), $G_{A^{i,e}}$ has no positive-weight path ending at $i$.
By condition (i), we have $v_{i}(A_{i}\cup\{e\})=v_{i}(A_{i})$.
Then, the weights of the arcs leaving $i$ do not change in $G_{A}$ and $G_{A^{i,e}}$.
Also, it is clear that arcs irrelevant to $i$ do not change their weights.
On the other hand, each $j\in N\setminus\{i\}$ satisfies $v_{j}(A_{i}\cup\{e\})\in\{v_{j}(A_{i}), v_{j}(A_{i})+1\}$.

Take any path $P$ in $G_{A^{i,e}}$. If it does not contain $i$, then its weight is the same as that in $G_{A}$ and is at most 1. 
If a path contains $i$, divide $P$ into $P'$ and $P''$, where the former is a subpath from the initial node to $i$ and 
the latter is a subpath from $i$ to the last node.
Since $G_{A^{i,e}}$ has no positive-weight path ending at $i$, the weight of $P'$ is non-positive. 
Moreover, the weight of $P''$ does not differ in $G_A$ and $G_{A^{i.e}}$, and hence is at most 1 also in $G_{A^{i.e}}$.
Hence, the weight of $P$ is at most 1 in $G_{A^{i.e}}$.  
\end{proof}

By condition (ii) in Lemma~\ref{lem:complement2} and Theorem~\ref{thm:envy-graph}, the allocation with a subsidy returned by the SEC algorithm is 
envy-free, with each agent receiving a subsidy of $0$ or $1$.
Furthermore, there is at least one agent $i\in N$ such that $G_A$ has no path of weight $1$ starting at $i$
(since otherwise there exists a positive-weight cycle in $G_A$, which contradicts (ii)).
Thus, the total subsidy is at most $n-1$.
By the algorithm and condition (i) in Lemma~\ref{lem:complement2}, 
the allocation is complete and utilitarian optimal.
To complete the proof of Theorem~\ref{thm:matroidrank:complete}, we show the following claim,
which is needed to demonstrate that the algorithm does not fall into an infinite loop at Step 2 (b).

\begin{lemma}\label{lem:complement}
In Step 2, for each item $e$, any agent is chosen as $i$ in (b) at most once. 
Hence, (b) is repeated at most $n$ times.
\end{lemma}
\begin{proof}
For any agent chosen in Step 2 (b), the weight of a path $P_i$ is at least $1$ in $G_{A^{i,e}}$.
By the argument in the proof of Lemma~\ref{lem:complement2}, then its weight in $G_A$ is at least $0$ (i.e., we have $w(P_i)\geq 0$, where $w$ denotes the weight function with respect to $G_A$).
Moreover, note that $G_A$ has no positive-weight cycle by (ii) in Lemma~\ref{lem:complement2}.

Suppose, to the contrary, that some agent is chosen in Step (b) multiple times.
Without loss of generality, let $1$ be such an agent and suppose that 
$2,3,\dots,k$ appear in this order between the first and the second appearance of $1$.
Then each $P_i~(i=1,2,\dots,k-1)$ is a path from an agent $i+1$ to $i$, and $P_k$ is a path from $1$ to $k$.
By connecting paths $P_i~(i=1,2,\dots,k)$, we can obtain a directed walk $Q$ that starts and ends at $1$.
If some arcs are used in multiple paths, then we replace them with multi-arcs so that each arc is used exactly once in $Q$.
Each multi-arc has the weight same as that of the original arc.
By $w(P_i)\geq 0~(i=1,2,\dots,k)$, we have $w(Q)=\sum_{i=1}^k w(P_i)\geq 0$; for each node, the indegree and outdegree in $Q$ coincide.
Then, the walk $Q$ is partitioned into a family $\mathcal{C}$ of directed cycles,
satisfying $\sum_{C\in \mathcal{C}}w(C)=w(Q)\geq 0$.
Since each path has a non-positive weight in $G_{A}$, every cycle in $\mathcal{C}$ has weight $0$.
Then we obtain $w(P_i)=0~(i=1,2,\dots,k)$. 

Fix any $i^*\in \{1,2,\dots,k\}$ and let $j^*$ be the second last node in $P_{i^*}$; that is, $(j^*,i^*)$ is the last arc in $P_{i^*}$.
Since $P_{i^*}$ has a positive weight in $G_{A^{i^*,e}}$ while $w(P_{i^*})=0$ in $G_A$,
we see that the weight of $(j^*,i^*)$ in $G_{A^{i^*,e}}$ is larger than its weight in $G_{A}$ by $1$,
i.e., $v_{j^*}(A_{i^*}\cup\{e\})-v_{j^*}(A_{i^*})=1$.
Among the cycles in $\mathcal{C}$, let $C^*$ be the one containing $(j^*,i^*)$.
Let us define $A'\coloneqq(A'_1, A'_2,\dots,A'_n)$ as follows:
$A'_{j^*}\coloneqq A_{i^*}+e$, $A'_j\coloneqq A_i~(\forall(j,i)\in C^*:j\neq j^*)$, $A'_j\coloneqq A_j~(\forall j\in N\setminus C^*)$.
Since $C^*$ has weight $0$ in $G_A$, we have $\sum_{(j,i)\in C^*}(v_j(A_i)-v_j(A_j))=0$.
Using this, we obtain
\begin{align*}
\sum_{j\in N}v_j(A'_j)-\sum_{j\in N}v_j(A_j)
=&~\sum_{j\in C^*}\bigl(v_j(A'_j)-v_j(A_j)\bigr)\\
=&\sum_{(j,i)\in C^*}\bigl(v_j(A_i)-v_j(A_j)\bigr)+\bigl(v_{j^*}(A_{i^*}\cup\{e\})-v_{j^*}(A_{i^*})\bigr)
=1.
\end{align*}
Then, the utilitarian social welfare of $A'$ is strictly larger than that of $A$, contradicting condition~(i) in Lemma~\ref{lem:complement2}.
\end{proof}

We now show that the SEC algorithm runs in polynomial time.
Step 1 of this algorithm is the same as that of the SE mechanism and hence can be computed in polynomial time by Lemma~\ref{lem:SE_poly}.
Moreover, Steps 2 and 3 can be computed in polynomial time by the method used by Halpern and Shah \cite{halpern2020fair}, i.e., by applying the Floyd-Warshall algorithm to  the graph obtained by negating all arc weights in the envy graph.

We remark that, as shown in Appendix (Proposition~\ref{prop:EFX}), the allocation returned by the SEC algorithm is EFX.
Thus, we obtain the following result: 
\begin{corollary}
For matroidal valuations, there exists a complete allocation that is utilitarian optimal and EFX.
\end{corollary}

\section{Superadditive valuations}\label{sec:superadditive}
In this section, we consider a class of valuations that do not possess the substitution property, namely, a class of superadditive valuations. Holmstr\"{o}m \cite{holmstrom:econometrica:1979} proved that
when the set $V$ of valuations satisfies the \emph{convexity} condition, the Groves mechanisms are the only utilitarian optimal and truthful mechanisms. For superadditive valuations,
some instances of the Groves mechanisms, including 
the VCG mechanism, satisfy envy-freeness \cite{papai:scw:2003}.
The class of superadditive valuations also satisfies convexity.
Thus, according to Holmstr\"{o}m \cite{holmstrom:econometrica:1979}, 
the Groves mechanisms are the only utilitarian optimal and truthful
mechanisms for such valuations.

We require that the subsidy for each agent must be non-negative; to fulfill this goal, we can use the following mechanism:
\begin{tcolorbox}[title=VCG with an upfront subsidy $m$, left=0mm]
\begin{enumerate}[label=\textbf{Step \arabic*.},leftmargin=*]
%
\item Allocate items according to an arbitrarily chosen $A^* \in \arg\max_{A} \sum_{j \in N} v_j(A_j)$.
\item Give $m - \bigl(\max_{A} \sum_{j \neq i} v_j(A_j) - \sum_{j\neq i} v_j(A^*_j)\bigr)$ subsidy to each $i \in N$.
\end{enumerate}
\end{tcolorbox}
\begin{theorem}
For superadditive valuations, the VCG with an upfront subsidy $m$ is truthful, utilitarian
optimal, and envy-free, and each subsidy is in $[0,m]$. 
\end{theorem}
\begin{proof}
By definition, the resulting allocation is utilitarian optimal. 
Note that the second term of the subsidy (i.e., 
$\max_{A} \sum_{j \neq i} v_j(A_j) - \sum_{j\neq i} v_j(A^*_j)$) is
equal to the standard VCG payment.
Thus, this mechanism is equivalent to the following
mechanism; first, each agent
obtains an upfront subsidy $m$. Then, items are allocated using the 
standard VCG, where each agent pays the VCG payment from the upfront
subsidy. By distributing the same amount of upfront subsidy for each
agent, the overall mechanism still satisfies envy-freeness and
truthfulness.
Moreover, the standard VCG payment is non-negative and at most
$v_i(A^*_i)$. Since we assume that $v_i(M) \leq m$ holds, the subsidy is
non-negative and at most $m$. 
\end{proof}

We note that, for additive valuations, one can compute a utilitarian optimal allocation in polynomial time.\footnote{This can be done by allocating each item to the agent who likes the most.} Hence, the above mechanism is polynomial-time implementable for a class of additive valuations. However, generally, the problem is NP-hard for superadditive valuations (see, e.g., Proposition 11.5 of \cite{Nisan}). 

Now, can the amount of subsidy for each agent be reduced while achieving envy-freeness and utilitarian optimality? The next theorem shows that the required subsidy for each agent is in fact $m$ even when there are two agents with additive valuations.
\begin{theorem}
\label{thm:impossibility1}
For any $\epsilon>0$,
if a mechanism is envy-free and utilitarian optimal, 
it requires a subsidy of $m-\epsilon$ for each agent,
even when there are only two agents with additive valuations such that the value of each item is at most $1$.
\end{theorem}
\begin{proof}
Suppose that there are two agents $N=\{1,2\}$ with valuation functions $v_1(X)=|X|$ and $v_2(X)=(1-\epsilon/m)|X|$ for each $X\subseteq M$.
Then, the unique utilitarian optimal allocation is $A=(M,\emptyset)$, and the mechanism must pay $m-\epsilon$ subsidy to agent $2$.
\end{proof}

\section{General monotone valuations}\label{sec:general}
In this section, we consider a class of monotone valuations. 
P\'{a}pai \cite{papai:scw:2003} showed that for general monotone valuations,
no instance of the Groves mechanisms \cite{groves:econometrica:1973}
satisfies envy-freeness; hence, for monotone valuations, there exists no mechanism satisfying utilitarian optimality, truthfulness, and envy-freeness. 
This negative result has been strengthened into a class of monotone submodular valuations by Feldman and Lai~\cite{Feldman2012}.\footnote{Specifically, Theorem \ref{thm:impossibility:general} applies to a subclass of monotone submodular valuations called \emph{capacitated valuations} \cite{Feldman2012}.}

\begin{theorem}[Feldman and Lai~\cite{Feldman2012}]\label{thm:impossibility:general}
No mechanism satisfies truthfulness, envy-freeness, and utilitarian optimality, even when all agents have monotone submodular valuations. 
\end{theorem}

If we require completeness instead of utilitarian optimality, we can construct a mechanism that satisfies truthfulness and envy-freeness; in fact, Caragiannis and Ioannidis~\cite{caragiannis2020computing} pointed out that the following mechanism satisfies these properties: allocate all items to the agent $i^*$ that values it the most and pay subsidy of $v_{i^*}(M)$ to every other agent.
Note that the subsidy for each agent is at most $m$ by the assumption that the maximum valuation is bounded by $m$. On the other hand, a complete and envy-free mechanism requires each agent to receive a subsidy of $m$.\footnote{Note that we do not know whether a similar example exists under the assumption that the maximum marginal contribution of each item is $1$ for each agent.}
\begin{theorem}\label{thm:general:complete}
If a mechanism satisfies completeness and envy-freeness,
then it requires a subsidy of $m$ for each agent, even when there are two agents.
\end{theorem}
\begin{proof}
Suppose that there are $m$ items $M=\{e_1,\dots,e_m\}$ and two agents $N=\{1,2\}$ with valuation functions
\begin{align}
    v_1(X)=v_2(X)=
    \begin{cases}
    m    & \text{if }e_1\in X,\\
    0    & \text{if }e_1\not\in X.
    \end{cases}
\end{align}
By completeness, one agent receives $e_1$ and the other agent does not.
Without loss of generality, we assume that agent $1$ receives $e_1$; then, by envy-freeness, agent $2$ must be subsidized by at least $m$.
\end{proof}


\section{Conclusion}
We have studied the mechanism design for allocating an indivisible resource with a limited amount of subsidy. 
Although it is difficult in general to provide any theoretical guarantees, we identified that a class of matroidal valuations does admit a desired mechanism using a subsidy of at most $1$ for each agent. For superadditive valuations, we showed that there is a truthful mechanism that is both envy-free and utilitarian optimal and that requires a subsidy of $m$ for each agent.

There remain several questions left open. Although our work is primarily concerned with utilitarian optimality as an efficiency criterion, it would be interesting to study the compatibility of truthfulness and fairness with other efficiency requirements, such as completeness and non-wastefulness. 
As a specific question, the VCG mechanism with an upfront subsidy can allocate all items and achieves the bound of $m$ for additive valuations. An obvious direction would be to study whether the amount of $m$ is necessary to achieve a truthful, envy-free, and complete mechanism when agents have additive valuations, assuming that the maximum value of each item is $1$. 
Another important topic is to understand how truthful and complete mechanisms with limited subsidy look like. In the mechanism design without money, Amanatidis et al.~\cite{Amanatidis} characterized such mechanisms with concerning two agents. It remains a challenge to extend the result of \cite{Amanatidis} to the setting with a subsidy.  
We also highlight that, for general monotone valuations beyond matroidal and additive valuations, it remains an open question as to what is the asymptotically minimal amount of subsidies required to make some allocation envy-free. Brustle et al.~\cite{Brustle2020} showed that, for monotone valuations, an envy-free allocation with subsidy $2(n-1)$ for each agent exists, assuming that the maximum marginal contribution of each item is $1$ for each agent; however, it is unclear whether this bound is tight. 

\clearpage

\bibliographystyle{plain}

\clearpage
\appendix

\section{Fairness notions for indivisible goods}

The literature on indivisible goods is concerned with relaxed fairness notions, since an envy-free allocation may not exist.
Here we summarize the fairness notions for an allocation $A$ of indivisible goods:
\begin{description}
    \item[Envy-free (EF):] $A$ is called envy-free if $v_i(A_i) \geq v_i(A_j)$ for all agents $i,j\in N$.
    \item[Envy-free up to any good (EFX):] $A$ is called envy-free up to any good if for any agents $i,j \in N$, either $v_i(A_i) \geq v_i(A_j)$ or for every $e \in A_j$ we have $v_i(A_i) \geq v_i(A_j \setminus \{e\})$. 
    \item[Envy-free up to one good (EF1):] $A$ is called envy-free up to one good if for any agents $i,j \in N$, either $v_i(A_i) \geq v_i(A_j)$ or there exists $e \in A_j$ such that $v_i(A_i) \geq v_i(A_j \setminus \{e\})$,
\end{description}

By definition, EF implies EFX, and EFX implies EF1.
In Section~\ref{sec:matroidal}, we provided a mechanism and an algorithm, 
each of which returns an allocation that is envy-freeable with subsidy at most 1 for each agent. 
This property implies that the amount of envy some agent has towards other is at most 1. 
For matroidal valuations, which have dichotomous marginals, this condition immediately implies that the returned allocation is EF1.
Further, as shown by Babaioff et al.~\cite[Proposition 5]{babaioff2020fair}, any clean Lorenz dominating allocation is EFX,%
\footnote{Here we briefly introduce the proof of Babaioff et al.~\cite[Proposition 5]{babaioff2020fair}.
Let $A$ be a clean Lorenz dominating allocation and suppose that $v_i(A_i)<v_i(A_j)$ for some $i,j\in N$. 
As $A$ is clean, we have $v_i(A_i)=|A_i|$. If $|A_j|\leq |A_i|+1$, then we have $v_i(A_j)=|A_j|=v_i(A_i)+1$, and hence the EFX condition holds.
If $|A_j|\geq |A_i|+2$, then by the cleanness of $A$, we have $v_j(A_j)=|A_j|\geq|A_i|+2=v_i(A_i)+2$.
As $v_i(A_i)<v_i(A_j)$, the augmentation property of matroid implies that there is $e\in A_j$ such $v_i(A_i\cup\{e\})=v_i(A_i)+1$
(and $v_j(A_j\setminus\{e\})=v_j(A_j)-1$).
Since $v_j(A_j)=v_i(A_i)+2$, by moving $e$ from $A_j$ to $A_i$, we obtain an allocation that Lorenz dominates $A$, a contradiction.}
and so is the allocation returned by the SE mechanism given in Section~\ref{sec:matroidal}.
We will show in Proposition~\ref{prop:EFX} that the output of the SEC algorithm is also EFX.

\subsection{The SEC algorithm returns an EFX allocation}
\begin{proposition}\label{prop:EFX}
For matroidal valuations, the allocation returned by the SEC algorithm is EFX.
\end{proposition}
\begin{proof}
For matroidal valuations, when $v_i(A_i)<v_i(A_j)$ is assumed, the condition $v_i(A_i)\geq v_i(A_j\setminus\{e\})$ 
holds only if $v_i(A_j\setminus\{e\})=v_i(A_j)-1=v_i(A_i)$. 
Further, $v_i(A_j\setminus\{e\})=v_i(A_j)-1$ holds for every $e\in A_j$ 
only if $A_j$ is independent with respect to $v_i$, i.e., $v_i(A_j)=|A_j|$.
Therefore, the EFX condition is equivalent to the following condition for matroidal valuations.
\begin{description}
\item[\rm ($\star$)] For any arc $(i,j)$ with positive weight in $G_A$, we have $v_i(A_j)=|A_j|=v_i(A_i)+1$.
\end{description}
Just after Step 1, $A$ is a clean Lorenz dominating allocation, and hence is EFX as shown in Babaioff et al.~\cite[Proposition 5]{babaioff2020fair}.
So ($\star$) is satisfied. We now show that, this condition is preserved whenever $A$ is updated.

Let $i^*\in N$ be an agent who receives an item $e$ in Step 2 (c).
We show that ($\star$) holds for the allocation $A'\coloneqq A^{i^*,e}$.
By Step 2 (b), $G_{A'}$ has no positive-weight path ending at $i^*$.
In particular, $G_{A'}$ has no arc that enters $i^*$ and has weight 1.
Take some arc $(i,j)$ that has a positive weight in $G_{A'}$. Then $j\neq i^*$ by the above argument. 
Hence $(i,j)$ has the same weight as in $G_A$. (This is shown in the proof of Lemma~\ref{lem:complement2}.)
Then $(i,j)$ has a positive weight also in $G_A$ and the condition ($\star$) for $A$ implies $v_i(A_j)=|A_j|=v_i(A_i)+1$. 
Since $j\neq i^*$, we have $A'_j=A_j$. Further, $v_i(A_i)=v_i(A'_i)$ regardless of whether $i=i^*$ or not 
(in case $i=i^*$, it follows from the utilitarian optimality of $A$). Thus, the condition ($\star$) is preserved.
\end{proof}

\end{document}